\documentclass[journal]{IEEEtran}
\ifCLASSINFOpdf
\else
\fi

\usepackage{amsmath,amsfonts,amssymb,color}
\usepackage[dvipsnames]{xcolor}
\definecolor{winered}{rgb}{0.8,0,0}
\definecolor{myblue}{rgb}{0,0,0.8}
\definecolor{mygreen}{rgb}{0.0 0.47 0.44}
\usepackage{tikz}
\usepackage{pmat}
\usepackage{amsthm}
\usepackage{empheq}

\newtheorem{definition}{Definition}
\newtheorem{theorem}{Theorem}
\newtheorem{lemma}{Lemma}
\newtheorem{proposition}{Proposition}
\newtheorem{corollary}{Corollary}
\newtheorem{remark}{Remark}

\newtheorem{example}{Example}
\usepackage{epsfig} % for postscript graphics files

\DeclarePairedDelimiter{\floor}{\lfloor}{\rfloor}

\DeclareMathOperator*{\argmin}{\arg\!\min}

\usepackage{algorithm}
\usepackage{algpseudocode}

\usepackage{epsfig}
\usepackage{array}
\usepackage{multirow}
\usepackage{epstopdf}
\usepackage{tikz}
\usepackage{relsize}
\usetikzlibrary{shapes,arrows}
\usepackage[hidelinks]{hyperref}
\hypersetup{
    colorlinks=true,
    linkcolor={winered},
    citecolor={myblue}
}
\usepackage{cite}

\begin{document}
\title{Distributed State Estimation over Time-Varying Graphs: Exploiting the Age-of-Information} 

\author{Aritra Mitra, John A. Richards, Saurabh Bagchi and Shreyas Sundaram
% <-this % stops a space
\thanks{A. Mitra, S. Bagchi and S. Sundaram are with the School of Electrical and Computer Engineering at Purdue University. J. A.  Richards is with Sandia National Laboratories.   Email: {\tt \{mitra14, sbagchi, sundara2\}@purdue.edu},  {\tt{jaricha@sandia.gov}}. This work was supported in part by NSF CAREER award
1653648, and by the Laboratory Directed Research and Development program at Sandia National Laboratories. Sandia National Laboratories is a multimission laboratory managed and operated by National Technology \& Engineering Solutions of Sandia, LLC, a wholly owned subsidiary of Honeywell International Inc., for the U.S. Department of Energy's National Nuclear Security Administration under contract DE-NA0003525. The views expressed in the article do not necessarily represent the views of the U.S. Department of Energy or the United States Government.}}

\maketitle
\begin{abstract}
    We study the problem of designing a distributed observer for an LTI system over a time-varying communication graph. The limited existing work on this topic imposes various restrictions either on the observation model or on the sequence of communication graphs. In contrast, we propose a  single-time-scale distributed observer that works under mild assumptions. Specifically, our communication model only requires strong-connectivity to be preserved over non-overlapping, contiguous intervals that are even allowed to grow unbounded over time. We show that under suitable conditions that bound the growth of such intervals, joint observability is sufficient to track the state of any discrete-time LTI system exponentially fast, at any desired rate. In fact, we also establish finite-time convergence based on our approach. Finally, we develop a variant of our algorithm that is provably robust to worst-case adversarial attacks, provided the sequence of graphs is sufficiently connected over time. The key to our approach is the notion of a  ``freshness-index" that keeps track of the age-of-information being diffused across the network. Such indices enable nodes to reject stale estimates of the state, and, in turn, contribute to stability of the error dynamics. 
\end{abstract}
\section{Introduction} Given a discrete-time LTI system $\mathbf{x}[k+1]=\mathbf{Ax}[k]$, and a linear measurement model $\mathbf{y}[k]=\mathbf{Cx}[k]$, a classical result in control theory states that one can design an observer that generates an asymptotically correct estimate $\hat{\mathbf{x}}[k]$ of the state $\mathbf{x}[k]$, if and only if the pair $(\mathbf{A},\mathbf{C})$ is detectable \cite{chen}. Additionally, if the pair $(\mathbf{A},\mathbf{C})$ is observable, then one can achieve exponential convergence at any desired convergence rate. Over the last couple of decades, significant effort has been directed towards studying the distributed counterpart of the above problem, wherein observations of the process are distributed among a set of sensors modeled as nodes of a communication graph  \cite{dist3,ugrinov,kim,martins,mitraTAC,wang,han,rego,nozal,wang2,ren}. A fundamental question that arises in this context is as follows:  What are the minimal requirements on the measurement structure of the nodes and the underlying communication graph that guarantee the existence of a distributed observer? Here, by a distributed observer, we imply a set of state estimate update and information exchange rules that enable each node to track the entire state asymptotically. The question posed above was answered only recently in \cite{martins,mitraTAC,wang,han,rego,nozal} for static graphs. However, when the underlying network changes with time, very little is known about the distributed state estimation problem - a gap that we intend to bridge in this paper. Our study is motivated by  applications in environmental monitoring tasks using mobile robot teams \cite{env1,env2}, where time-varying communication graphs arise naturally either as a consequence of robot-mobility, or due to packet drops and link failures typical in wireless sensor networks.  

% In particular, our goal is to build a firm theoretical understanding of the following question: \textit{How infrequently can the nodes in a sensor network communicate with each other, and yet, collaboratively track the state of an unstable, linear time-invariant dynamical system?} Our interest in this problem is not merely academic, but one that is spurred by various realistic scenarios. Imagine a network of drones deployed over a region for the purpose of tracking the evolution of a diffusive process \cite{env1,env2}. The mobility of such drones, or packet drops and link failures typical in wireless sensor networks,  intermittently disrupt the process of information exchange. To improve the battery-life of sensors, one might even deliberately want to limit the number of transmissions by each sensor node. To this end, the design of communication-efficient algorithms are starting to gain prominence in a variety of contexts, namely distributed optimization \cite{comm1}, parameter estimation \cite{comm2}, hypothesis testing \cite{comm3}, and federated learning \cite{comm4}. In each of the scenarios discussed above, a natural way to  mathematically model the aspect of intermittent interactions is via a time-varying communication graph. 

\textbf{Related Work}: The existing literature on the design of  distributed observers can be broadly classified in terms of the assumptions made on the system model, the observation model, and the network structure. Recent works on this topic \cite{martins,mitraTAC,wang,han,rego,nozal} that make minimal system- and graph-theoretic assumptions can be further classified based on a finer set of attributes, as follows. (i) Does the approach require multiple consensus iterations between two consecutive time-steps of the dynamics?\footnote{Such approaches, referred to as two-time-scale approaches, incur high communication cost, and might hence be prohibitive for real-time applications.} (ii) What is the dimension of the estimator maintained by each node? (iii) Is the convergence asymptotic, or in finite-time? (iv) In case the convergence is asymptotic, can the convergence rate be controlled? The techniques proposed in \cite{martins,mitraTAC,wang,han,rego,nozal} operate on a single-time-scale, and those in \cite{mitraTAC,han,rego,nozal} require observers of dimension no more than that of the state of the system. The notion of convergence in each of the papers \cite{martins,mitraTAC,wang,han,rego,nozal} is asymptotic; the ones in \cite{wang,han,nozal} can achieve exponential convergence at \textit{any} desired rate. For dynamic networks, there is much less work. Two recent papers  \cite{wang2} and \cite{mitraAR} provide theoretical guarantees for certain specific classes of time-varying graphs; while the former proposes a two-time-scale approach, the latter develops a single-time-scale algorithm. However, the extent to which such results can be generalized has remained an open question.

\textbf{Contributions}: The main contribution of this paper is the development of a single-time-scale distributed state estimation algorithm in Section \ref{sec:algo} that requires each node to maintain an estimator of dimension equal to that of the state (along with some simple counters), and works under assumptions that are remarkably mild in comparison with those in the existing literature. Specifically, in terms of the observation model, we only require joint observability, i.e., the state is observable w.r.t. the collective observations of the nodes. This assumption is quite standard, and can be relaxed to joint detectability, with appropriate implications for the convergence rate. 

However, the key departure from existing literature comes from the generality of our communication model, introduced formally in Section \ref{sec:probform}. Based on this model, we require strong-connectivity to be preserved over non-overlapping, contiguous intervals that can even grow linearly with time at a certain rate. In other words, we allow the inter-communication intervals between the nodes to potentially grow unbounded. Even under the regime of such sparse communications, we establish that our proposed approach can track the state of any discrete-time LTI system exponentially fast, at \textit{any} desired convergence rate. While all the works on distributed state estimation that we are aware of provide an asymptotic analysis, estimating the state of the system in \textit{finite-time} might be crucial in certain safety-critical applications. To this end, we show that under a suitable selection of the observer gains, one can achieve finite-time convergence based on our approach. To put our results into context, it is instructive to compare them with the work closest to ours, namely \cite{wang2}. In \cite{wang2}, the authors study a continuous-time analog of the problem under consideration, and develop a solution that leverages an elegant connection to the problem of distributed linear-equation solving \cite{mou}. In contrast to our technique, the one in \cite{wang2} is inherently a two-time-scale approach, requires each node to maintain and update auxiliary state estimates, and works under the assumption that the communication graph is strongly-connected at every instant. 

Our work is also related to the vast literature on consensus \cite{jadcons} and distributed optimization \cite{nedic} over time-varying graphs, where one typically assumes  strong-connectivity to be preserved over uniformly bounded intervals - a special case of our communication model. It is important to recognize that, in contrast to these settings, the problem at hand requires tracking potentially unstable external dynamics, making the stability analysis considerably more challenging. In particular, one can no longer directly leverage convergence properties of products of stochastic matrices. From a system-theoretic point of view, the error dynamics under our communication model evolves based on a switched linear system, where certain modes are potentially unstable. Thus, one way to analyze such dynamics is by drawing on techniques from the switched system stability literature \cite{lin}. However, such an analysis can potentially obscure the role played by the network structure. Instead, we directly  exploit the interplay between the structure of the changing graph patterns on the one hand, and the evolution of the estimation error dynamics on the other. Doing so, we provide a comprehensive stability analysis of our estimation algorithm employing simple, self-contained arguments from linear system theory and graph theory. 

Finally, in Section \ref{sec:advtimevar}, we extend our study of distributed state estimation over dynamic graphs to settings where certain nodes are under attack, and can act arbitrarily. For the problem under consideration, accounting for adversarial behavior is highly non-trivial even for static graphs, with only a few recent results that provide any theoretical guarantees \cite{deghat,junsoo,he,mitraAuto}. Accounting for worst-case adversarial behavior over time-varying networks is considerably harder, since the compromised nodes can not only transmit incorrect, inconsistent state estimates, but also lie about the time-stamps of their data. Nonetheless, we develop a novel algorithm to handle such scenarios, and establish its correctness when the graph-sequence is sufficiently connected over time.

The key idea behind our approach is the use of a suitably defined ``freshness-index" that keeps track of the age-of-information being diffused across the network. Loosely speaking, the ``freshness-index"  of a node quantifies the extent to which its estimate of the state is outdated. Exchanging such indices enables a node to assess, in real-time, the quality of the estimate of a neighbor. Accordingly, it can reject estimates that are relatively stale - a fact that contributes to the stability of the error dynamics. We point out that while this is perhaps the first use of the notion of age-of-information (AoI) in a networked control/estimation setting, such a concept has been widely  employed in the study of various queuing-theoretic problems arising in wireless networks \cite{kaul,costa,talak}.\footnote{The notion of age-of-information (AoI) was first introduced in \cite{kaul} as a performance metric to keep track of real-time status updates in a communication system. In a wireless network, it measures the time elapsed since the generation of the packet most recently delivered to the destination. In Section \ref{sec:algo}, we will see how such a concept applies to the present setting.}

To sum up, we propose the first single-time-scale distributed state estimation algorithm that provides both finite-time and exponentially fast convergence guarantees,  under significantly milder assumptions on the time-varying graph sequences than those in the existing literature. In addition, we show how our algorithm can be extended to allow for worst-case adversarial attacks as well. 

A preliminary version of this paper appeared as \cite{ACC19}. Here, we significantly expand upon the content in \cite{ACC19} in two main directions. First, we generalize the results in \cite{ACC19} to allow for unbounded inter-communication intervals. Second, the extension to adversarial scenarios in Section \ref{sec:advtimevar} is a new addition entirely. We also provide full proofs of all our results. 

\section{Problem Formulation and Background}
\label{sec:probform}
\textbf{System and Measurement Model}: We are interested in collaborative state estimation of a discrete-time LTI system of the form:
\begin{equation}
\mathbf{x}[k+1]=\mathbf{A}\mathbf{x}[k],
\label{eqn:system}
\end{equation}
where $k\in\mathbb{N}$ is the discrete-time index, $\mathbf{A}\in\mathbb{R}^{n \times n}$ is the system matrix, and $\mathbf{x}[k]\in\mathbb{R}^{n}$ is the state of the system.\footnote{We use $\mathbb{N}$ and $\mathbb{N}_{+}$ to denote the set of non-negative integers and the set of positive integers, respectively.} A network of sensors, modeled as nodes of a communication graph, obtain partial measurements of the state of the above process as follows:
\begin{equation}
\mathbf{y}_{i}[k]=\mathbf{C}_i\mathbf{x}[k],
\label{eqn:Obsmodel}
\end{equation}
where $\mathbf{y}_{i}[k] \in {\mathbb{R}}^{r_i}$ represents the measurement vector of the $i$-th node at time-step $k$, and $\mathbf{C}_i \in {\mathbb{R}}^{r_i \times n}$ represents the corresponding observation matrix. Let $\mathbf{y}[k]={\begin{bmatrix}\mathbf{y}^T_{1}[k] & \cdots & \mathbf{y}^T_{N}[k]\end{bmatrix}}^T$ and $\mathbf{C}={\begin{bmatrix}\mathbf{C}^T_{1} & \cdots & \mathbf{C}^T_{N}\end{bmatrix}}^T$ represent the collective measurement vector at time-step $k$, and the collective observation matrix, respectively. The goal of each node $i$ in the network is to generate an asymptotically correct estimate $\hat{\mathbf{x}}_i[k]$ of the true dynamics $\mathbf{x}[k]$. It may not be possible for any node $i$ in the network to accomplish this task in isolation, since the pair $(\mathbf{A},\mathbf{C}_i)$ may not be detectable in general. Throughout the paper, we will only assume that the pair $(\mathbf{A},\mathbf{C})$ is observable; the subsequent developments can be readily generalized to the case when $(\mathbf{A},\mathbf{C})$ is detectable.

\textbf{Communication Network Model}: As is evident from the above discussion, information exchange among nodes is necessary for all nodes to estimate the full state. At each time-step $k\in\mathbb{N}$, the available communication channels are modeled by a directed communication graph $\mathcal{G}[k]=(\mathcal{V},\mathcal{E}[k])$, where $\mathcal{V}=\{1,\ldots,N\}$ represents the set of nodes, and $\mathcal{E}[k]$ represents the edge set of $\mathcal{G}[k]$ at time-step $k$. Specifically, if $(i,j)\in\mathcal{E}[k]$, then node $i$ can send information directly to node $j$ at time-step $k$; in such a case, node $i$ will be called a neighbor of node $j$ at time-step $k$. We will use $\mathcal{N}_i[k]$ to represent the set of  all neighbors (excluding node $i$) of node $i$ at time-step $k$. When $\mathcal{G}[k]=\mathcal{G}  \hspace{2mm}\forall k\in\mathbb{N}$, where $\mathcal{G}$ is a static, directed communication graph, the necessary and sufficient condition (on the system and network) to solve the distributed state estimation problem is the joint detectability of each source component of $\mathcal{G}$ \cite{martins}.\footnote{A source component of a static, directed graph is a strongly connected component with no incoming edges.} Our \textbf{goal} in this paper is to extend the above result to  scenarios where the underlying communication graph is allowed to change over time. To this end, let the union graph over an interval $[k_1,k_2], 0 \leq k_1 \leq k_2$, denoted $\bigcup\limits_{\tau=k_1}^{k_2}\mathcal{G}[\tau]$,  indicate a graph with vertex set equal to $\mathcal{V}$, and  edge set equal to the union of the edge sets of the individual graphs appearing over the interval $[k_1,k_2]$. Based on this convention, we now formally describe the communication patterns (induced by the sequence $\{G[k]\}^{\infty}_{k=0}$) that are considered in this paper. We assume that there exists a sequence $\mathbb{I}=\{t_0,t_1,\ldots\}$ of increasing time-steps with $t_0=0$ and each $t_i\in\mathbb{N}$, satisfying the following conditions.
\begin{enumerate}
    \item[(C1)] Define the mapping $f:\mathbb{I}\rightarrow\mathbb{N}_{+}$ as $f(t_q)=t_{q+1}-t_{q}, \forall t_q\in\mathbb{I}.$ We assume that $f(t_q)$ is a non-decreasing function of its argument.
    \item[(C2)] For each $k\in\mathbb{N}$, let $m(k)\triangleq \max\{t_q\in\mathbb{I}:t_q \leq k\}$, and $M(k)\triangleq \min\{t_q\in\mathbb{I}:t_q > k\}$. Define $g:\mathbb{N}\rightarrow\mathbb{N}_{+}$ as $g(k)=M(k)-m(k)=f(m(k)).$ Then, we assume that the following holds:
    \begin{equation}
        \limsup_{k\to\infty}\frac{2(N-1)g(k)}{k} = \delta < 1. 
        \label{eqn:ass_growth}
    \end{equation}
    \item[(C3)] For each $t_q\in\mathbb{I}$, we assume that  $\bigcup\limits_{\tau=t_q}^{t_{q+1}-1}\mathcal{G}[\tau]$ is strongly-connected. 
\end{enumerate}

Let us discuss what the above conditions mean. Condition (C1) in conjunction with condition (C3) tells us that the intervals over which strong-connectivity is preserved are non-decreasing in length (evident from the monotonicity of the function $f(\cdot)$).\footnote{Our results can be generalized to account for the case when $f(\cdot)$ is non-monotonic by suitably modifying condition (C2).} Essentially, our aim here is to come up with distributed estimators that function correctly despite sparse communication; hence, we allow for potentially growing inter-communication intervals. Since we place no assumptions at all on the spectrum of the $\mathbf{A}$ matrix, stability of the estimation error process imposes natural restrictions on how fast the inter-communication intervals can be allowed to grow. Condition (C2) formalizes this intuition by constraining such intervals to grow at most linearly at a certain rate. Notably, conditions (C1)-(C3) cover a very general class of time-varying graph sequences. In particular, we are unaware of any other work that allows the inter-communication intervals to grow unbounded for the problem under consideration. 

Consider the following two examples. (i) The mapping $f$ satisfies $f(t_q)=T, \forall t_q\in\mathbb{I}$, where $T$ is some positive integer. (ii) The mapping $f$ satisfies $f(t_q)=\floor{\sqrt{t_q+1}}, \forall t_q\in\mathbb{I}$. It is easy to verify that (C2) is satisfied in each case. While example (i) represents the standard ``joint strong-connectivity" setting where the inter-communication intervals remain bounded, example (ii) deviates from existing literature by allowing the inter-communication intervals to grow unbounded. 

\textbf{Background}: For communication graphs satisfying conditions (C1)-(C3) as described above, our \textbf{objective} will be to design a distributed algorithm that ensures $\lim_{k\to\infty}\left\Vert\hat{\mathbf{x}}_i[k]-\mathbf{x}[k]\right\Vert=0, \forall i\in\mathcal{V}$, with $\hat{\mathbf{x}}_i[k]$ representing the estimate of the state $\mathbf{x}[k]$ maintained by node $i\in\mathcal{V}$. To this end, we recall the following result from \cite{mitraTAC}.
\begin{lemma}
\label{transformations}
Given a system matrix $\mathbf{A}$, and a set of $N$ sensor observation matrices $\mathbf{C}_1, \mathbf{C}_2, \ldots, \mathbf{C}_{N}$, define $\mathbf{C} \triangleq {\begin{bmatrix}\mathbf{C}^T_{1} & \cdots & \mathbf{C}^T_{N}\end{bmatrix}}^T$. Suppose $(\mathbf{A},\mathbf{C})$ is observable. Then, there exists a similarity transformation matrix $\mathbf{T}$ that transforms the pair $\mathbf{(A,C)}$ to $(\bar{\mathbf{A}},\bar{\mathbf{C}})$, such that
\begin{equation}
\resizebox{0.7\hsize}{!}{$
\begin{aligned}
\bar{\mathbf{A}} &= \left[
\begin{array}{c|c|cc}
\mathbf{A}_{11}  & \multicolumn{3}{c}{\mathbf{0}} \\
\hline
\mathbf{A}_{21}
 & 
\mathbf{A}_{22} & \multicolumn{2}{c}{\mathbf{0}} \\
\cline{2-4}
\vdots & \vdots & \hspace{-5mm} \ddots  & \vdots \\
\mathbf{A}_{N1}& \mathbf{A}_{N2} \hspace{2mm}\cdots & \mathbf{A}_{N(N-1)} & \multicolumn{1}{|c}{\mathbf{A}_{NN}} 
\end{array}
\right], \\~\\
\bar{\mathbf{C}} &= \begin{bmatrix} \bar{\mathbf{C}}_1 \\ \bar{\mathbf{C}}_2 \vspace{-1mm} \\  \vdots \\ \bar{\mathbf{C}}_N \end{bmatrix} = \left[ \begin{array}{cccc} \mathbf{C}_{{11}} & \multicolumn{3}{|c}{\mathbf{0}}\\
\hline
\mathbf{C}_{{21}} & \multicolumn{1}{c}{\mathbf{C}_{{22}}} & \multicolumn{2}{|c}{\mathbf{0}}\\
\hline
\vdots&\vdots&\vdots&\vdots\\
\mathbf{C}_{{N1}} & \multicolumn{1}{c}{\mathbf{C}_{{N2}}}  & \cdots  \mathbf{C}_{N(N-1)} & \mathbf{C}_{NN}\\
 \end{array}
 \right],
\end{aligned}
$}
\label{eqn:gen_form}
\end{equation}
and the pair $(\mathbf{A}_{ii},\mathbf{C}_{ii})$ is observable $\forall i \in \{1,2, \ldots, N\}$.
\end{lemma}

We use the matrix $\mathbf{T}$ given by Lemma \ref{transformations} to perform the coordinate transformation $\mathbf{z}[k]=\mathbf{T}^{-1}\mathbf{x}[k]$,  yielding:
\begin{equation}
\begin{aligned}
\mathbf{z}[k+1]&=\bar{\mathbf{A}}\mathbf{z}[k],\\
\mathbf{y}_i[k]&=\bar{\mathbf{C}}_i\mathbf{z}[k], \quad \forall i \in \{1, \ldots, N\},
\end{aligned}
\label{eqn:coordinatetransform}
\end{equation}
where $\bar{\mathbf{A}}={\mathbf{T}}^{-1}\mathbf{A}\mathbf{T}$ and $\bar{\mathbf{C}}_i = \mathbf{C}_i\mathbf{T}$ are given by (\ref{eqn:gen_form}). Commensurate with the structure of $\bar{\mathbf{A}}$, the vector $\mathbf{z}[k]$ is of the following form:
\begin{equation}
\mathbf{z}[k]={\begin{bmatrix}
{\mathbf{z}^{(1)}_{}[k]}^{T}&
\cdots&
{\mathbf{z}^{(N)}_{}[k]}^{T}
\end{bmatrix}}^{T},
\label{eqn:sub-states}
\end{equation}
where $\mathbf{z}^{(j)}[k]$ will be referred to as the $j$-th sub-state. By construction, since the pair $(\mathbf{A}_{jj},\mathbf{C}_{jj})$ is locally observable w.r.t. the measurements of node $j$, node $j$ will be viewed as the unique source node for sub-state $j$. In this sense, the role of node $j$ will be to ensure that each non-source node $i\in\mathcal{V}\setminus\{j\}$ maintains an asymptotically correct estimate of sub-state $j$. For a time-invariant strongly-connected graph, this is achieved in \cite{mitraTAC} by first constructing a spanning tree rooted at node $j$, and then requiring nodes to only listen to their parents in such a tree for estimating sub-state $j$. The resulting  unidirectional flow of information from the source $j$ to the rest of the network guarantees stability of the error process for sub-state $j$ \cite{mitraTAC}.

However, the above strategy is no longer applicable when the underlying communication graph is time-varying, for the following reasons. (i) For a given sub-state $j$, there may not exist a common spanning tree rooted at node $j$ in each graph $\mathcal{G}[k], k\in\mathbb{N}$. (ii) Assuming that a specific spanning tree rooted at node $j$ is guaranteed to repeat at various points in time (not necessarily periodically), is restrictive, and qualifies as only a special case of conditions (C1)-(C3). (iii) Suppose for simplicity that $\mathcal{G}[k]$ is strongly-connected at each time-step (as in \cite{wang2}), and hence, there exists a spanning tree $\mathcal{T}_j[k]$ rooted at node $j$ in each such graph. For estimating sub-state $j$, suppose consensus at time-step $k$ is performed along the spanning tree $\mathcal{T}_j[k]$. As we demonstrate in the next section, switching between such spanning trees can lead to unstable error processes over time. Thus, if one makes no further assumptions on the system model beyond joint observability, or on the sequence of communication graphs beyond conditions (C1)-(C3), ensuring stability of the estimation error dynamics becomes a challenging proposition. Nonetheless, we develop a simple algorithm in Section \ref{sec:algo} that bypasses the above problems. In the next section, we provide an example that helps to build the intuition behind this algorithm. 

\section{An Illustrative Example}
\label{sec:example}
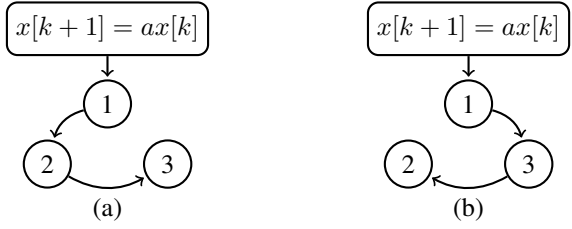
\begin{figure}[t]
\begin{center}
%\hspace*{-1cm}
\begin{tikzpicture}
[->,shorten >=1pt,scale=.4, minimum size=5pt, auto=center, node distance=2cm,
  thick,  node/.style={circle, draw=black, thick},]
\tikzstyle{block} = [rectangle, draw, text centered, rounded corners, minimum height=0.7cm, minimum width=0.5cm];
\node [block]  at (-2,2.5) (plant1) {$x[k+1]=ax[k]$};
\node [circle, draw](n1) at (-2,0)  (1)  {1};
\node [circle, draw](n2) at (-4,-2)   (2)  {2};
\node [circle, draw](n3) at (0,-2)   (3)  {3};
\node [ ] ( ) at (-2,-3.5) () {(a)};
\path[every node/.style={font=\sffamily\small}]
(plant1) 
        edge [] node [] {} (1)

    (1)
        edge [bend right] node [] {} (2)
        
    (2)
        edge [bend right] node [] {} (3);
        
\node [block]  at (10,2.5) (plant2) {$x[k+1]=ax[k]$};   
\node [circle, draw](n1) at (10,0)     (4)  {1};
\node [circle, draw](n2) at (8,-2)   (5)  {2};
\node [circle, draw](n3) at (12,-2)   (6)  {3};
\node [ ] ( ) at (10,-3.5) () {(b)};
\path[every node/.style={font=\sffamily\small}]
(plant2) 
        edge [] node [] {} (4)

(4)
       edge [bend left]  node [] {} (6)
        
 (6)   edge [bend left]  node [] {} (5);
       
\end{tikzpicture}
\end{center}
\caption{An LTI system is monitored by a network of 3 nodes, where the communication graph $\mathcal{G}[k]$ switches between the two graphs shown above.}
\label{fig:example}
\end{figure}
\begin{figure}[t]
\begin{center}
\begin{tabular}{cc}
\includegraphics[scale=0.095]{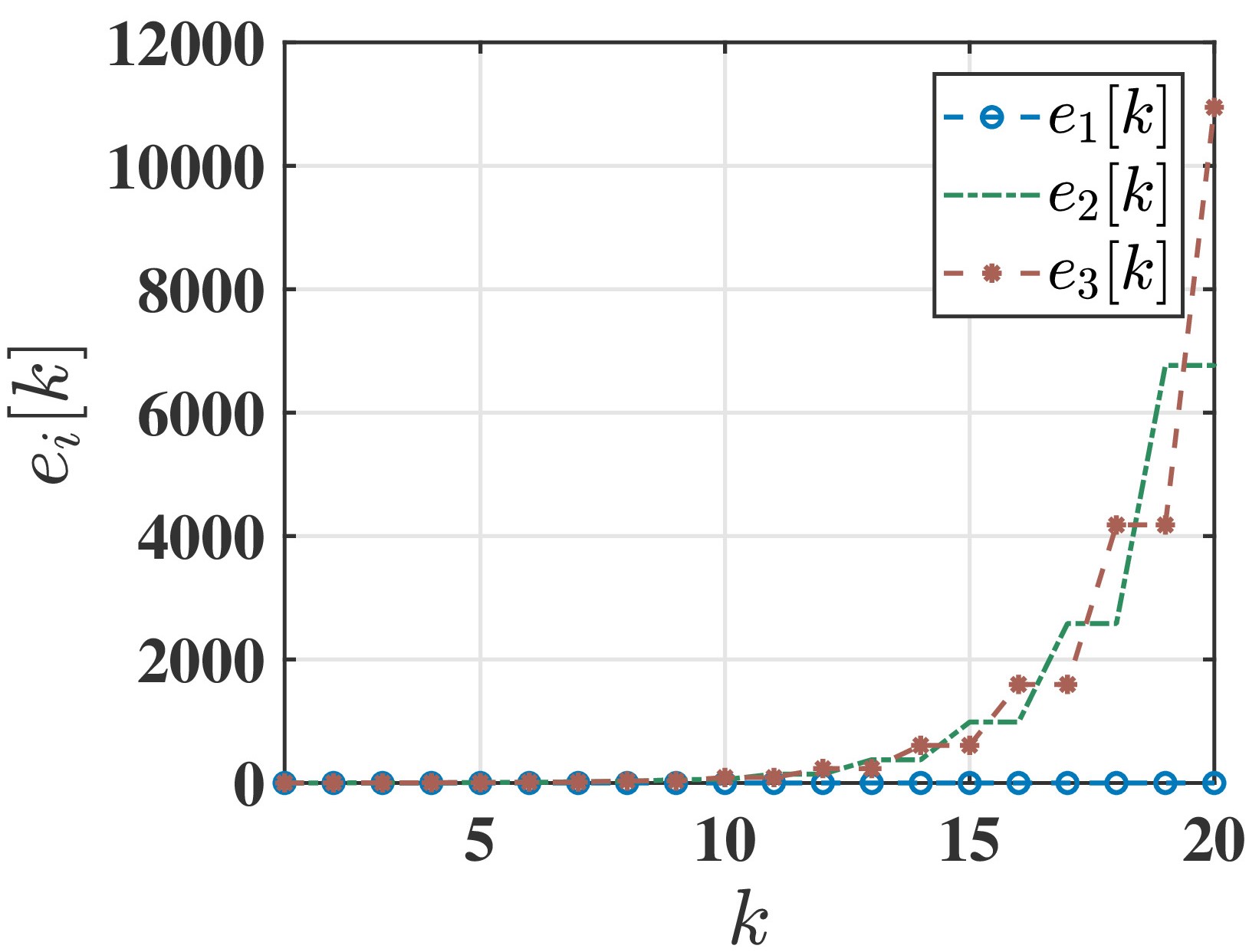}&\includegraphics[scale=0.095]{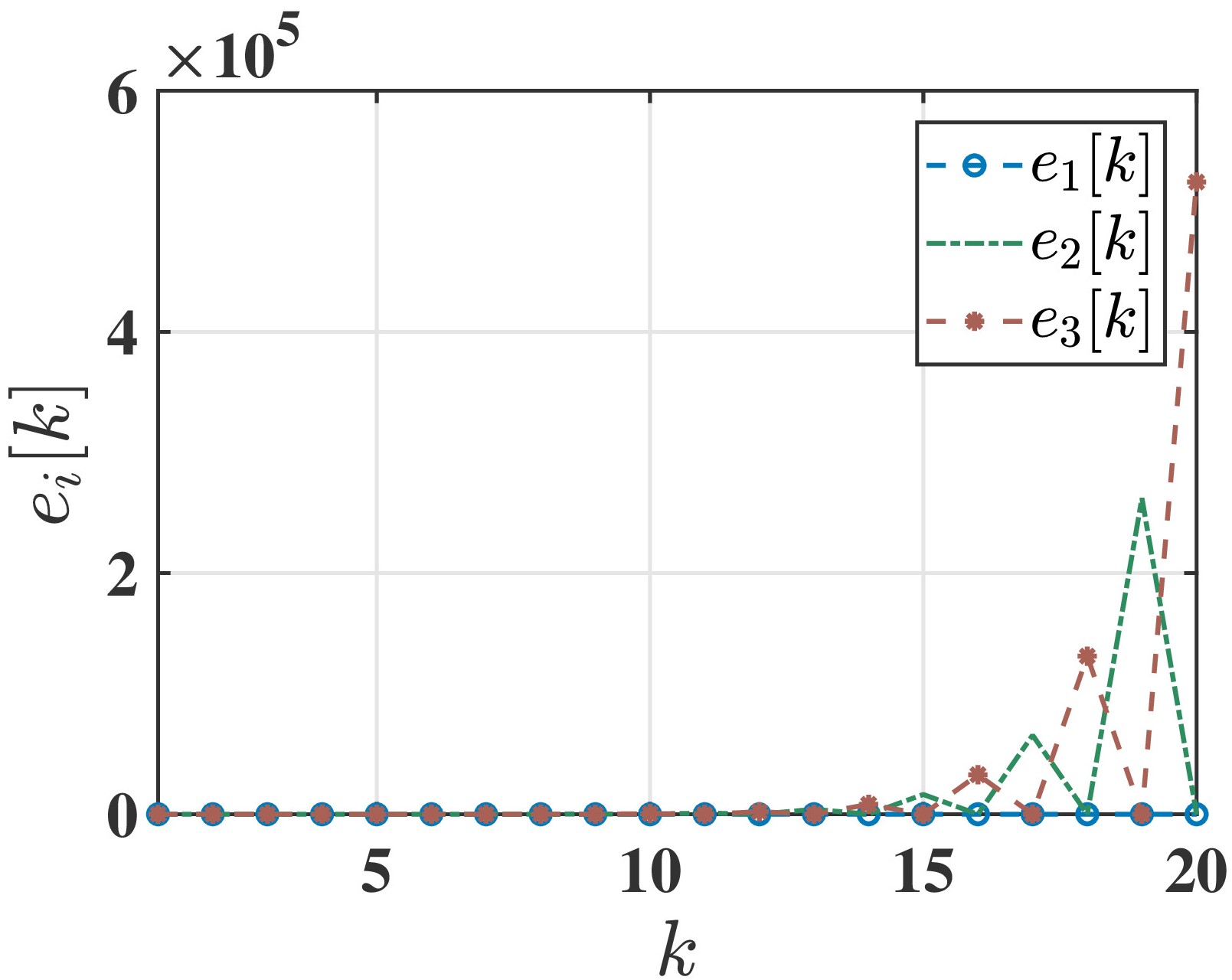}
\end{tabular}
\end{center}
\caption{Estimation error plots of the nodes for the model in Figure \ref{fig:example}. Simulations are performed for a model where $a=2$. The figure on the left corresponds to the case where consensus weights are distributed uniformly among neighbors, while the one on the right is the case where weights are placed along a tree rooted at node 1.}
\label{fig:error}
\end{figure}
Consider a network of 3 nodes monitoring a scalar unstable process $x[k+1]=ax[k]$, as shown in Figure \ref{fig:example}. The communication graph $\mathcal{G}[k]$ switches between the two topologies shown in Figure \ref{fig:example}. Specifically, $\mathcal{G}[k]$ is the graph in Figure $\ref{fig:example}$(a) at all even time-steps, and the one in $\ref{fig:example}$(b) at all odd time-steps. Node 1 is the only node with non-zero measurements, and thus acts as the source node for this network. Suppose for simplicity that it perfectly measures the state at all time-steps, i.e., its state estimate is $\hat{x}_1[k]=x[k], \forall k\in\mathbb{N}$. Given this setup, a standard consensus based state estimate update rule would take the form (see for example \cite{martins,mitraTAC,wang2}):
\begin{equation}
    \hat{x}_i[k+1]=a\left(\sum_{j\in\mathcal{N}_i[k]\cup\{i\}}w_{ij}[k]\hat{x}_j[k]\right), i\in\{2,3\},
    \label{eqn:sampleupdate}
\end{equation}
where the weights $w_{ij}[k]$ are non-negative, and satisfy $\sum_{j\in\mathcal{N}_i[k]\cup\{i\}}w_{ij}[k]=1, \forall k\in\mathbb{N}$. The key question is: how should the consensus weights be chosen to guarantee stability of the estimation errors of nodes 2 and 3? Even for this simple example, if such weights are chosen naively, then the errors may grow unbounded over time. To see this, consider the following two choices: (1) consensus weights are distributed evenly over the set $\mathcal{N}_i[k]\cup\{i\}$, and (2) consensus weights are placed along the tree rooted at node $1$. In each case, the error dynamics are unstable, as depicted in Figure \ref{fig:error}. To overcome this problem, suppose nodes 2 and 3 are aware of the fact that node 1 has perfect information of the state. Since nodes 2 and 3 have no measurements of their own, intuitively, it makes sense that they should place their consensus weights entirely on node 1 whenever possible. The trickier question for node 2 (resp., node 3) is to decide \textit{when} it should listen to node 3 (resp., node 2). Let us consider the situation from the perspective of node 2. At time-step 0, it adopts the information of node 1, and hence, the error of node 2 is zero at time-step 1. However, the error of node 3 is not necessarily zero at time-step 1. Consequently, if node 2 places a non-zero consensus weight on the estimate of node 3 at time-step 1, its error at time-step 2 might assume a non-zero value. Clearly, at time-step 1, node 2 is better off rejecting the information from node 3, and simply running open-loop. \textit{The main take-away point here is that adoption or rejection of information from a neighbor should be based on the quality of information that such a neighbor has to offer}. In particular, a node that has come in contact with node 1 more recently is expected to have better information about the state than the other. Thus, to dynamically evaluate the quality of an estimate, the above reasoning suggests the need to introduce a metric that keeps track of how recent that estimate is with respect to  (w.r.t.) the estimate of the source node 1. In the following section, we formalize this idea by introducing such a metric.
\begin{algorithm}[t]
	\caption{}
	\label{algo:main}
	\begin{algorithmic}[1]
	\State \textbf{Initialization:} $\tau^{(j)}_j[0]=0, \tau^{(j)}_i[0]=\omega, \forall i\in\mathcal{V}\setminus\{j\}$.
	\State \textbf{Update Rules for the Source Node:} 
	Node $j$ maintains $\tau^{(j)}_j[k]=0, \forall k\in\mathbb{N}$. It updates $\hat{\mathbf{z}}^{(j)}_j[k]$ as:
	\begin{align}
     \hat{\mathbf{z}}^{(j)}_j[k+1]={\mathbf{F}}_{jj}\hat{\mathbf{z}}^{(j)}_j[k]+\sum \limits_{q=1}^{(j-1)}\mathbf{G}_{jq}\hat{\mathbf{z}}^{(q)}_j[k]
     +\mathbf{L}_j\mathbf{y}_j[k],
     \label{eqn:sourceupdate}
    \end{align}
where ${\mathbf{F}}_{jj}=(\mathbf{A}_{jj}-\mathbf{L}_j\mathbf{C}_{jj})$, $\mathbf{G}_{jq}=(\mathbf{A}_{jq}-\mathbf{L}_j\mathbf{C}_{jq})$, and $\mathbf{L}_j$ is an observer gain to be designed later.
\State \textbf{Update Rules for the Non-Source Nodes:}  Each non-source node $i\in\mathcal{V}\setminus\{j\}$ operates as follows. 
\State \underline{\textbf{Case 1: $\tau^{(j)}_i[k]=\omega$}}. Define $\mathcal{M}^{(j)}_i[k]\triangleq\{l\in\mathcal{N}_i[k]: \tau^{(j)}_l[k]\neq\omega\}.$ If $\mathcal{M}^{(j)}_i[k]\neq\emptyset$, let $u\in\argmin_{l\in\mathcal{M}^{(j)}_i[k]} \tau^{(j)}_l[k]$. Node $i$ updates $\tau^{(j)}_i[k]$ and $\hat{\mathbf{z}}^{(j)}_i[k]$ as:
\begin{equation}
    \tau^{(j)}_i[k+1]=\tau^{(j)}_u[k]+1,
    \label{eqn:indexupdatecase11}
\end{equation}
\vspace{-6mm}
\begin{equation}
    \hat{\mathbf{z}}^{(j)}_i[k+1]=\mathbf{A}_{jj}\hat{\mathbf{z}}^{(j)}_u[k]+\sum \limits_{q=1}^{(j-1)}\mathbf{A}_{jq}\hat{\mathbf{z}}^{(q)}_i[k].
    \label{eqn:nonsourcecase1}
\end{equation}
If $\mathcal{M}^{(j)}_i[k]=\emptyset$, then
\begin{equation}
    \tau^{(j)}_i[k+1]=\omega,
    \label{eqn:indexupdatecase12}
\end{equation}
\vspace{-6mm}
\begin{equation}
    \hat{\mathbf{z}}^{(j)}_i[k+1]=\mathbf{A}_{jj}\hat{\mathbf{z}}^{(j)}_i[k]+\sum \limits_{q=1}^{(j-1)}\mathbf{A}_{jq}\hat{\mathbf{z}}^{(q)}_i[k].
  \label{eqn:nonsourcecase2}
\end{equation}
\State \underline{\textbf{Case 2: $\tau^{(j)}_i[k]\neq\omega$}}. Define 
    $\mathcal{F}^{(j)}_i[k]\triangleq\{l\in\mathcal{M}^{(j)}_i[k]: \tau^{(j)}_l[k] < \tau^{(j)}_i[k]\}$, where $\mathcal{M}^{(j)}_i[k]$ is as defined in line 4. If $\mathcal{F}^{(j)}_i[k]\neq\emptyset$, let $u\in\argmin_{l\in\mathcal{F}^{(j)}_i[k]} \tau^{(j)}_l[k]$. Node $i$ then updates $\tau^{(j)}_i[k]$ as per \eqref{eqn:indexupdatecase11}, and  $\hat{\mathbf{z}}^{(j)}_i[k]$ as per \eqref{eqn:nonsourcecase1}. If $\mathcal{F}^{(j)}_i[k]=\emptyset$, then $\tau^{(j)}_i[k]$ is updated as
\begin{equation}
    \tau^{(j)}_i[k+1]=\tau^{(j)}_i[k]+1,
    \label{eqn:indexupdatecase2}
\end{equation}
and $\hat{\mathbf{z}}^{(j)}_i[k]$ is updated as per \eqref{eqn:nonsourcecase2}.
	\end{algorithmic}
\end{algorithm}

\section{Algorithm}
\label{sec:algo}
Building on the intuition developed in the previous section, we introduce a new approach to designing distributed observers for a general class of time-varying networks. The main idea is the use of a ``freshness-index"  that keeps track of how delayed the estimates of a node are w.r.t. the estimates of a source node. Specifically, for updating its estimate of $\mathbf{z}^{(j)}[k]$, each node $i\in\mathcal{V}$ maintains and updates at every time-step a freshness-index $\tau^{(j)}_i[k]$. At each time-step $k\in\mathbb{N}$, the index $\tau^{(j)}_i[k]$ plays the following role: it determines whether node $i$ should adopt the information received from one of its neighbors in $\mathcal{N}_i[k]$, or run open-loop, for updating $\hat{\mathbf{z}}^{(j)}_i[k]$, where $\hat{\mathbf{z}}^{(j)}_i[k]$ represents the estimate of $\mathbf{z}^{(j)}[k]$ maintained by node $i$. In case it is the former, it also indicates which specific neighbor in $\mathcal{N}_i[k]$  node $i$ should listen to at time-step $k$; this piece of information is particularly important for the problem under consideration, and ensures stability of the error process. The rules that govern the updates of the freshness indices $\tau^{(j)}_i[k]$, and the estimates of the $j$-th sub-state $\mathbf{z}^{(j)}[k]$, are formally stated in Algorithm \ref{algo:main}. In what follows, we describe each of these rules.

 {\textbf{Discussion of Algorithm \ref{algo:main}:}} Consider any sub-state $j\in\{1,\ldots,N\}.$ Each node $i\in\mathcal{V}$ maintains an index $\tau^{(j)}_i[k]\in\{\omega\}\cup\mathbb{N}$, where $\omega$ is a dummy value. Specifically, $\tau^{(j)}_i[k]=\omega$  represents an ``infinite-delay"  w.r.t. the estimate of the source node for sub-state $j$, namely node $j$, i.e., it represents that node $i$ has not received any information (either directly or indirectly) from node $j$ regarding sub-state $j$ up to time-step $k$. For estimation of sub-state $j$, since delays are measured w.r.t. the source node $j$, node $j$ maintains its freshness-index $\tau^{(j)}_j[k]$ at zero for all time, to indicate a zero delay w.r.t. itself. For updating its estimate of $\mathbf{z}^{(j)}[k]$, it uses only its own information, as is evident from \eqref{eqn:sourceupdate}. 
 
 Every other node starts out with an ``infinite-delay" $\omega$ w.r.t. the source (line 1 of Algo. \ref{algo:main}). The freshness-index of a node $i\in\mathcal{V}\setminus\{j\}$ changes from $\omega$ to a finite value when it comes in contact with a neighbor with a finite delay, i.e., with a freshness-index that is not $\omega$ (line 4 of Algo. \ref{algo:main}). At this point, we say that $\tau^{(j)}_i[k]$ has been ``triggered". Once triggered, at each time-step $k$, a non-source node $i$ will adopt the information of a neighbor $l \in \mathcal{N}_i[k]$ only if node $l$'s estimate of $\mathbf{z}^{(j)}[k]$ is ``more fresh" relative to its own, i.e., only if $\tau^{(j)}_l[k] < \tau^{(j)}_i[k]$.\footnote{Under Case 1 or Case 2 in Algo.  \ref{algo:main}, when a node $i\in\mathcal{V}\setminus\{j\}$ updates $\tau^{(j)}_i[k]$ via \eqref{eqn:indexupdatecase11}, and $\hat{\mathbf{z}}^{(j)}_i[k]$ via \eqref{eqn:nonsourcecase1}, we say that ``$i$ adopts the information of $u$ at time $k$ for sub-state $j$"; else, if it runs open-loop, we say it adopts its own information.} Among the set of neighbors in $\mathcal{M}^{(j)}_i[k]$ (if $\tau^{(j)}_i[k]$ has not yet been triggered), or in $\mathcal{F}^{(j)}_i[k]$ (if $\tau^{(j)}_i[k]$ has been triggered), node $i$ only adopts the information (based on \eqref{eqn:nonsourcecase1}) of the neighbor $u$ with the least delay. At this point, the delay of node $i$ matches that of node $u$, and this fact is captured by the update rule \eqref{eqn:indexupdatecase11}. In case node $i$ has no neighbor that has fresher information than itself w.r.t. sub-state $j$ (where informativeness is quantified by $\tau^{(j)}_i[k])$, it increments its own freshness-index by $1$ (as per \eqref{eqn:indexupdatecase2}) to capture the effect of its own information getting older, and runs open-loop based on \eqref{eqn:nonsourcecase2}. Based on the above rules, at any given time-step $k$, $\tau^{(j)}_i[k]$ measures the age-of-information of $\hat{\mathbf{z}}^{(j)}_i[k]$, relative to the source node $j$. This fact is established later in Lemma \ref{lemma:form}. Finally, note that Algorithm \ref{algo:main} describes an approach for estimating $\mathbf{z}[k]$, and hence $\mathbf{x}[k]$, since $\mathbf{x}[k]=\mathbf{Tz}[k]$.
\section{Performance Guarantees For Algorithm \ref{algo:main}}
\label{sec:results_statements}
\subsection{Statement of the Results}
In this section, we first state the theoretical guarantees afforded by Algorithm \ref{algo:main}, and then discuss their implications; proofs of these statements are deferred to Appendix \ref{subsec:proofs} and \ref{app:proofProp1}. We begin with one of the main results of this paper. 
\begin{theorem}
Given an LTI system  \eqref{eqn:system}, and a measurement model \eqref{eqn:Obsmodel}, suppose  $(\mathbf{A,C})$ is observable. Let the sequence of communication graphs $\{\mathcal{G}[k]\}_{k=0}^{\infty}$ satisfy conditions (C1)-(C3) in Section \ref{sec:probform}. Then, given any desired convergence rate $\rho\in (0,1)$, the observer gains $\mathbf{L}_1, \ldots, \mathbf{L}_{N}$ can be designed in a manner such that the estimation error of each node converges to zero exponentially fast at rate $\rho$, based on Algorithm \ref{algo:main}.
\label{thm:main1}
\end{theorem}

The next result states that under the conditions in Theorem  \ref{thm:main1}, one can in fact achieve finite-time convergence. 
\begin{proposition} (\textbf{Finite-Time Convergence})
Suppose the conditions stated in Theorem \ref{thm:main1} hold. Then, the observer gains $\mathbf{L}_1, \ldots, \mathbf{L}_{N}$ can be designed in a manner such that the estimation error of each node converges to zero in finite-time.
\label{prop:finitetime}
\end{proposition}

The next result follows directly from Prop.   \ref{prop:finitetime} and indicates that, when the sequence of communication graphs exhibits certain structure, one can derive a closed form expression for the maximum number of time-steps required for convergence.

\begin{corollary}
Suppose the conditions stated in Theorem \ref{thm:main1} hold. Additionally, suppose $f(t_q) \leq T, \forall t_q\in\mathbb{I}$, where $T\in\mathbb{N}_{+}$. Then, the observer gains $\mathbf{L}_1, \ldots, \mathbf{L}_{N}$ can be designed in a manner such that the estimation error of each node converges to zero in at most $n+2N(N-1)T$ time-steps.
\label{corr:finitetime}
\end{corollary}

Let us now discuss the implications of our results, and comment on certain aspects of our approach. 

\begin{remark}
The fact that a network of partially informed nodes can track the state of a dynamical system with arbitrarily large eigenvalues, over inter-communication intervals that can potentially grow unbounded, is non-obvious \textit{a priori}. Our results in Thm. \ref{thm:main1} and Prop. \ref{prop:finitetime} indicate that not only can this be done exponentially fast at any desired rate, it can also be done in finite-time. In contrast, the result  closest to ours \cite{wang2} assumes strong-connectivity at each time-step  - an assumption that is significantly stronger than what we make. 
\end{remark}

\begin{remark} Notice that given a desired convergence rate $\rho$, the general design approach described in the proof of Thm. \ref{thm:main1} (in Appendix \ref{subsec:proofs}) offers a considerable degree of freedom in choosing the parameters $\lambda_1,\ldots,\lambda_N$;  the design flexibility so obtained in choosing the observer gains can be exploited to optimize transient performance, or performance against noise. In contrast, the proof of Prop.  \ref{prop:finitetime}  highlights a specific approach to obtain finite-time convergence. However, such an approach may lead to undesirable transient spikes in the estimation errors, owing to large observer gains.
\end{remark}
\begin{remark}
Presently, our approach requires a centralized design phase where the agents implement the multi-sensor decomposition in Section \ref{sec:probform}, and design their observer gains to achieve the desired convergence rate as outlined in  the proof of Thm. \ref{thm:main1}. This is of course a limitation, but one that is common to all existing approaches \cite{dist3,ugrinov,kim,martins,wang,ren,han,mitraTAC,rego,wang2,nozal} that we are aware of, i.e., each such approach involves a centralized design phase. Our approach also requires the nodes to know an upper-bound on the parameter $\delta$ in Eq. \eqref{eqn:ass_growth} while designing their observer gains. If, however, we constrain the inter-communication intervals to grow sub-linearly at worst, i.e, if $\delta=0$, then such gains can be designed with no knowledge on the nature of the graph-sequences. Thus, there exists a trade-off between the generality of the graph sequences that can be tolerated, and the information required to do so. 
\end{remark}
\begin{remark}
When strong-connectivity is preserved over uniformly bounded intervals, i.e., when $f(t_q) \leq T, \forall t_q\in\mathbb{I}$, for some $T\in\mathbb{N}_{+}$, then our approach leads to bounded estimation errors under bounded disturbances. However, as we will see in Section \ref{sec:disturbances}, this may no longer be the case if the inter-communication intervals grow unbounded, no matter how slowly.
\label{rem:boundederrs}
\end{remark}

\subsection{Implications of growing inter-communication intervals under bounded disturbances}
\label{sec:disturbances}
%Suppose an LTI system of the form \eqref{eqn:system} is excited by i.i.d. noise with bounded moments, or a bounded disturbance input. Then, in the centralized setting, input-to-state stability dictates that a standard Luenberger observer will lead to a bounded estimation error sequence, when the state is detectable via the measurements. An analogous result holds for the distributed case as well, when the underlying graph is time-invariant. Based on the convergence results in the previous section for unbounded communication intervals, one might then be prompted to think that bounded disturbances under such settings will manifest into bounded estimation errors. 
While Theorem \ref{thm:main1} shows that estimation is possible under growing inter-communication intervals, the goal of this section is to demonstrate via a simple example that this may no longer be true in the presence of disturbances. To this end, consider a scalar, unstable LTI system $x[k+1]=ax[k]+d$, where $a>1$, and $d>0$ is a disturbance input to the system. The network comprises of just 2 nodes: node 1 with measurement model $y_1[k]=c_1x[k], c_1\neq 0$, and node 2 with no measurements. Now consider an increasing sequence of time-steps $\mathbb{I}=\{t_0, t_1, \ldots \}$ with $t_0=0$, and let $f(t_q)=t_{q+1}-t_q, \forall t_q\in\mathbb{I}$ be a non-decreasing function of its argument, as in Section \ref{sec:probform}. Suppose the communication pattern comprises of an edge from node 1 to node 2 precisely at the time-steps given by $\mathbb{I}$. Node 1 maintains a standard Luenberger observer given by $\hat{x}_1[k+1]=a\hat{x}_1[k]+l_1(y_1[k]-c_1\hat{x}_1[k])$, where $l_1$ is the observer gain. Node 2 applies Algorithm \ref{algo:main}, which, in this case, translates to node 2 adopting the estimate of node 1 at each time-step $t_q$, and running open-loop at all other time-steps. Accordingly, we have $\hat{x}_2[t_{q+1}]=a^{f(t_q)}\hat{x}_1[t_q], \forall t_q\in\mathbb{I}$. With $e_i[k]=x[k]-\hat{x}_i[k], i\in\{1,2\}$, one can then easily verify:
\begin{equation}
\begin{aligned}
    e_1[k+1]&=\gamma e_1[k]+d, \forall k\in\mathbb{N},\\ 
    e_2[t_{q+1}]&=a^{f(t_q)}e_1[t_q]+d\frac{(a^{f(t_q)}-1)}{(a-1)}, \forall t_q\in\mathbb{I},
\end{aligned}
\label{eqn:unboundederr}
\end{equation}
where $\gamma=(a-l_1c_1)$. Now consider a scenario where the inter-communication intervals grow unbounded, i.e., $f(t_q)\rightarrow \infty$ as $t_q\rightarrow \infty$. Since $a>1$ and $d>0$, it is clear from \eqref{eqn:unboundederr} that the error subsequence $e_2[t_q],t_q\in\mathbb{I}$ will grow unbounded even if node 1 chooses $l_1$ such that $\gamma=0$. For the specific example under consideration, although the above arguments were constructed w.r.t. our algorithm, it seems unlikely that the final conclusion would change if one were to resort to some other approach.\footnote{Note that we are only considering single-time-scale algorithms where nodes are not allowed to exchange their measurements. Also, we assume here that the nodes have no knowledge about the nature of the disturbance $d$, thereby precluding the use of any disturbance-rejection technique.} The discussions in Section \ref{sec:results_statements} can be thus summarized as follows. 
\begin{itemize}
    \item For a noiseless, disturbance free LTI system of the form \eqref{eqn:system}, one can achieve exponential convergence at any desired rate, and even finite-time convergence based on Algorithm \ref{algo:main}, under remarkably mild assumptions: joint observability, and joint strong-connectivity over intervals that can potentially grow unbounded. 
    \item For an unstable system, any non-zero persistent disturbance, however small, can lead to unbounded estimation errors when the inter-communication intervals grow unbounded, no matter how slowly. Note however from Eq.  \eqref{eqn:unboundederr} that our approach leads to bounded estimation errors under bounded disturbances if the sequence $\{f(t_q)\}_{t_q\in\mathbb{I}}$ is uniformly bounded above (see Remark \ref{rem:boundederrs}).
\end{itemize}
In light of the above points, the reasons for stating our results in full generality, i.e., for unbounded communication intervals,  are as follows. First, we do so for theoretical interest, since we believe our work is the first to establish that the distributed state estimation problem can be solved with growing inter-communication intervals. Second, we essentially get this result for ``free", i.e., accounting for such general scenarios incurs no additional steps in terms of the design of Algorithm  \ref{algo:main}. Finally, we emphasize that, while no existing approach can even handle the case where strong-connectivity is preserved over uniformly bounded time intervals (i.e., $\exists \hspace{0.5mm}T\in\mathbb{N}_{+}$ such that $f(t_q) \leq T, \forall t_q\in\mathbb{I}$), the analysis for this scenario is simply a special case of that in Appendix  \ref{subsec:proofs}.
\section{Resilient Distributed State Estimation over Time-Varying Networks}
\label{sec:advtimevar}
We now consider a scenario where a subset of agents in the network is adversarial, and can deviate from the prescribed algorithm. We will show how one can employ the notion of freshness-indices to account for such adversarial agents over a time-varying network. To avoid cumbersome notation and to present the key ideas in a clear way, we will consider a scalar LTI system of the form $x[k+1]=ax[k]$. Later, we will discuss how our approach can be naturally extended to more general system models. We consider a  worst-case Byzantine adversary model, where an adversarial node is assumed to be omniscient, and can essentially act arbitrarily: it can transmit incorrect, potentially inconsistent information to its instantaneous out-neighbors, or choose not to transmit anything at all, while colluding with other adversaries in the process \cite{Byz,rescons,vaidyacons}. We will focus on an $f$-total adversarial model where the total number of adversaries in the network is bounded above by $f$, with  $f\in\mathbb{N}$. The adversarial set will be denoted by $\mathcal{A}$, and the set of regular nodes by $\mathcal{R}=\mathcal{V}\setminus\mathcal{A}$. Finally, for the scalar model under consideration, we will define the set of source nodes as follows: $\mathcal{S}=\{i\in\mathcal{V}: c_i\neq 0,\,\textrm{where}\, y_i[k]=c_ix[k]\}$, i.e., $i\in\mathcal{S}$ if and only if it can measure the state on its own. Note that we will allow $\mathcal{S}\cap\mathcal{A}\neq \emptyset$.
\subsection{Description of Algorithm \ref{algo:resilient}}
\begin{algorithm}[htbp]
	\caption{}
	\label{algo:resilient}
	\begin{algorithmic}[1]
	\State \textbf{Initialization:} For each $i\in\mathcal{V}\setminus\mathcal{S}$, $\tau_i[0]=\omega$; each entry of $\mathbf{v}_i[0],\mathbf{d}_i[0]$, and $\boldsymbol{\phi}_i[0]$ is empty; $\mathcal{M}_i[0]=\emptyset$; and $q_i[0]=0$. For each $i\in\mathcal{S}$, $\tau_i[0]=0$. 
   \State \textbf{Update Rules for Source Nodes:} 
	Each $i\in\mathcal{S}$ maintains $\tau_i[k]=0, \forall k\in\mathbb{N}$. It updates $\hat{x}_i[k]$ based on the following Luenberger observer, where $l_i$ is the observer gain:
	\begin{align}
     \hat{x}_i[k+1]=a\hat{x}_i[k]+l_i(y_i[k]-c_i\hat{x}_i[k]). 
     \label{eqn:algo2luen}
    \end{align}
\State \textbf{Update Rules for Non-Source Nodes:}  At each $k\in\mathbb{N}$, every non-source node $i\in\mathcal{V}\setminus\mathcal{S}$ operates as follows (lines 4-18). 
\State \underline{\textbf{Case 1: $\tau^{}_i[k]=\omega$}}. Define $\mathcal{J}_i[k] \triangleq \{j\in\mathcal{N}_i[k]: \tau_j[k] \neq \omega, \tau_j[k]\in\mathbb{N}, \tau_j[k] \leq k \}$. 
\If{$|\mathcal{J}_i[k]\setminus\mathcal{M}_i[k]| < (2f+1)-q_i[k]$}
\State Append each $l\in\mathcal{J}_i[k]\setminus\mathcal{M}_i[k]$ to $\mathcal{M}_i[k]$ \Comment{\textcolor{myblue}{This involves adding $l$ to $\mathcal{M}_i[k]$, and setting $v_{i,l}[k]=\hat{x}_l[k]$, $d_{i,l}[k]=\tau_l[k]$,  and $\phi_{i,l}[k]=k.$}}
\State Perform the following updates.
\begin{equation}
    q_i[k+1]=q_i[k]+|\mathcal{J}_i[k]\setminus\mathcal{M}_i[k]|.
    \label{eqn:c_update}
\end{equation}
\begin{equation}
    d_{i,l}[k+1]=d_{i,l}[k]+1, l\in\mathcal{M}_i[k].
\label{eqn:d_update}
\end{equation}
\begin{equation}
    \tau_i[k+1]=\omega.
\label{eqn:algo2tauomega}
\end{equation}
\begin{equation}
    \hat{x}_i[k+1]=a\hat{x}_i[k].
\label{eqn:algo2openloop}
\end{equation}
\Else
\State Sort the nodes in $\mathcal{J}_i[k]\setminus\mathcal{M}_i[k]$ in ascending order of their freshness-indices $\tau_l[k], l\in \mathcal{J}_i[k]\setminus\mathcal{M}_i[k] $. Append the first $(2f+1)-q_i[k]$ nodes in the resulting list to $\mathcal{M}_i[k]$. 
\State Set $q_i[\tau]=2f+1, \forall \tau \geq k+1$; update $d_{i,l}[k], l\in\mathcal{M}_i[k]$ as per \eqref{eqn:d_update}; update $\tau_i[k]$ as follows:
\begin{equation}
    \tau_i[k+1]= \max_{l\in\mathcal{M}_i[k]} d_{i,l}[k]+1.
\label{eqn:tau_update}
\end{equation}
\State For each $l\in\mathcal{M}_i[k]$, form the following quantity:
\begin{equation}
    \bar{x}_{i,l}[k]=a^{(k-\phi_{i,l}[k])}\hat{x}_l[\phi_{i,l}[k]].
    \label{eqn:barxil}
\end{equation}
\State Sort $\bar{x}_{i,l}[k],l\in\mathcal{M}_i[k]$ from highest to lowest, and reject the highest $f$ and the lowest $f$ of such quantities. Let the quantity that remains after such trimming be denoted $\bar{x}_i[k].$ Update $\hat{x}_i[k]$ as follows:
\begin{equation}
    \hat{x}_i[k+1]=a\bar{x}_i[k].
\label{eqn:algo2filterstep}
\end{equation}
\EndIf
\State Set $\mathcal{M}_i[k+1]=\mathcal{M}_i[k],\mathbf{v}_i[k+1]=\mathbf{v}_i[k]$, and $\boldsymbol{\phi}_i[k+1]=\boldsymbol{\phi}_i[k]$.
\State \underline{\textbf{Case 2: $\tau^{}_i[k]\neq\omega$}}. For each $l\in\mathcal{J}_i[k]\cap\mathcal{M}_i[k]$, if $\tau_l[k] < d_{i,l}[k]$, then set $v_{i,l}[k]=\hat{x}_l[k]$, $d_{i,l}[k]=\tau_l[k]$,  and $\phi_{i,l}[k]=k$ (we will classify this as an append operation). 
\State Sort the nodes in  $\mathcal{M}_i[k]\cup \{\mathcal{J}_i[k]\setminus\mathcal{M}_i[k]\}$ in ascending order of their freshness-indices, using 
 $d_{i,l}[k]$ as the index for $l\in\mathcal{M}_i[k]$, and  $\tau_l[k]$ as the index for $l\in\mathcal{J}_i[k]\setminus\mathcal{M}_i[k]$. Append the first $2f+1$ nodes in the resulting list to $\mathcal{M}_i[k].$ 
 \State For each $l\in\mathcal{M}_i[k]$, update $d_{i,l}[k]$ as per \eqref{eqn:d_update}. Update $\tau_i[k]$ as per \eqref{eqn:tau_update}, and $\hat{x}_i[k]$ via \eqref{eqn:algo2filterstep} based on the filtering operation described in lines 11-12. 
 \State Execute the operations in line 14. 
\end{algorithmic}
\end{algorithm}
Following the same line of reasoning as in Section \ref{sec:algo}, we are once again interested in answering the following question: When should a regular non-source node $i\in\{\mathcal{V}\setminus \mathcal{S}\}\cap\mathcal{R}$ use the information of a neighbor to update its estimate of the state? Unlike before, however, the presence of adversaries introduces certain complications: not only can an adversary transmit an arbitrary estimate of the state, it can also lie about its freshness-index. In particular, an adversary can follow the strategy of always reporting a freshness-index of $0$ so as to prompt its instantaneous out-neighbors to use its estimate value. To carefully account for such misbehavior, we devise a novel protocol, namely Algo. \ref{algo:resilient}. In what follows, we explain the main idea behind Algo. \ref{algo:resilient} in a nutshell.

\textbf{High-level idea}: Our goal is to ensure that the state estimate $\hat{x}_i[k]$ of a regular non-source node $i$ is ``close" to those of the regular source set $\mathcal{S}\cap\mathcal{R}$. To achieve this, we would like $\hat{x}_i[k]$ to be updated based on information that is neither too outdated, nor corrupted by the adversarial set. To meet the two requirements above, our main idea is to have node $i$ store and process information received from a dynamic list of neighbors. At each time-step, the list is first updated so as to retain only those nodes that have the most recent state estimates w.r.t. the source set. Subsequently, $\hat{x}_i[k]$ is updated by filtering out extreme estimates in the list in an appropriate manner. We now describe in detail how the above steps are implemented.

\textbf{Detailed Description}: Like Algo. \ref{algo:main}, Algo.  \ref{algo:resilient} requires each node $i$ to maintain a freshness-index $\tau_i[k]\in\mathbb{N}\cup\{\omega\}$. Each source node $i\in\mathcal{S}$ maintains $\tau_i[k]=0$ for all time, and updates $\hat{x}_i[k]$ based on a standard Luenberger observer (see Eq. \eqref{eqn:algo2luen}). Given the presence of adversarial nodes in the network, each regular non-source node $i$ relies on information redundancy for updating its estimate of the state. To achieve this, it  maintains four additional vectors $\mathbf{v}_i[k], \mathbf{d}_i[k]$, $\boldsymbol{\phi}_i[k]$, and $\mathcal{M}_i[k]$, each of dimension $2f+1$. The vector $\mathbf{v}_i[k]$ consists of state estimates received over time from $2f+1$ distinct nodes. The labels of such nodes are stored in a list  $\mathcal{M}_i[k]$, their freshness-indices in $\mathbf{d}_i[k]$, and  the time-stamps associated with their estimates (i.e., the time-step at which their estimate is received) are recorded in $\boldsymbol{\phi}_i[k]$. Algo. \ref{algo:resilient} requires each non-source node $i$ to execute two key steps:  (1) Maintaining a dynamical list $\mathcal{M}_i[k]$ of those $2f+1$ neighbors that have the lowest freshness-indices based on all the information node $i$ has acquired up to time $k$; and (2) Performing a filtering operation to update $\hat{x}_i[k]$ based on the latest state estimates of the nodes in the current list $\mathcal{M}_i[k]$. The second step, however, requires the cardinality of the set $\mathcal{M}_i[k]$ to be $2f+1$. Thus, Algo \ref{algo:resilient} involves an initial  pre-filtering phase (lines 4-7) where a non-source node simply gathers enough estimates to later act on.

$\bullet$ \textit{Discussion of Case 1}: We first describe the rules associated with the pre-filtering phase. Initially, each entry of $\mathbf{v}_i[0],\mathbf{d}_i[0]$, and $\boldsymbol{\phi}_i[0]$ is empty, and $\mathcal{M}_i[k]$ is an empty list (line 1). As time progresses, node $i$ adds distinct nodes to the list $\mathcal{M}_i[k]$ based on rules that we will discuss shortly. Until the time when  $|\mathcal{M}_i[k]|=2f+1$, node $i$ maintains $\tau_i[k]=\omega$, and uses a counter $q_i[k]$ to keep track of the number of entries in $\mathcal{M}_i[k]$. When $\tau_i[k]=\omega$, node $i$ operates as follows. It first considers the subset of neighbors $\mathcal{J}_i[k]$ at time $k$ that have freshness-indices other than $\omega$, belonging to $\mathbb{N}$, and at most $k$ (line 4).\footnote{Note that in the absence of adversaries, any node $j$ following Algo. \ref{algo:main} would satisfy $\tau_j[k]\in\mathbb{N}$ and $\tau_j[k] \leq k$, whenever $\tau_j[k]\neq \omega$. We would like the same to hold for any regular node following Algo. \ref{algo:resilient}. Thus, any neighbor reporting otherwise need not be considered for inclusion in $\mathcal{M}_i[k].$} In line 5 of Algo. \ref{algo:resilient}, node $i$ checks whether there are enough new nodes (i.e., nodes different from those already in $\mathcal{M}_i[k]$) in $\mathcal{J}_i[k]$ so as to bring the number of distinct entries in $\mathcal{M}_i[k]$ up to $(2f+1)$. If not, it ``appends" each node $l\in\mathcal{J}_i[k]\setminus\mathcal{M}_i[k]$ to $\mathcal{M}_i[k]$. By this, we mean that it adds $l$ to $\mathcal{M}_i[k]$, sets $v_{i,l}[k]=\hat{x}_l[k]$, $d_{i,l}[k]=\tau_l[k]$,  and $\phi_{i,l}[k]=k$ (see line 6). Here, we use the double subscript $i,l$ to indicate the $i$-th node's record of the various quantities associated with a node $l\in\mathcal{M}_i[k]$. Node $i$ keeps track of the number of new nodes appended via \eqref{eqn:c_update}. The entry $d_{i,l}[k]$ is the $i$-th node's internal copy of the freshness-index of node $l$, which it updates via \eqref{eqn:d_update} for future comparisons (such as those in lines 15 and 16). As an indicator of the fact that it has not yet acquired $2f+1$ state estimates to perform a filtering operation, node $i$ sets $\tau_i[k+1]=\omega$ (see Eq. \eqref{eqn:algo2tauomega}), and runs open-loop via \eqref{eqn:algo2openloop}. 

Let us now discuss the case when $\tau_i[k]$ transitions from $\omega$ to some value other than $\omega$ (lines 8-12). At the transition time-step, node $i$ appends those nodes from $\mathcal{J}_i[k]\setminus\mathcal{M}_i[k]$ to $\mathcal{M}_i[k]$ that have the lowest freshness-indices (see line 9). It does so in a way such that $|\mathcal{M}_i[k]|$ is precisely $2f+1$. Since node $i$ has now acquired enough estimates to perform the filtering step, it maintains $q_i[\tau]=2f+1$ for all $\tau \geq k+1$ (see line 10). As in the non-adversarial setting, the freshness-index of a node $i$ is a measure of how delayed its estimate is w.r.t. that of the source set $\mathcal{S}$. Since node $i$'s estimate in turn is updated based on the estimates of nodes in $\mathcal{M}_i[k]$, its freshness-index is essentially dictated by the largest entry in $\mathbf{d}_i[k]$, i.e., the entry corresponding to the most delayed estimate. This facet is captured by \eqref{eqn:tau_update}. Lines 11-12 of Algo. \ref{algo:resilient} constitute the filtering step that a node $i\in\mathcal{V}\setminus\mathcal{S}$ employs to update $\hat{x}_i[k]$. Since the latest estimates of nodes in $\mathcal{M}_i[k]$ might have different time-stamps, we need a way to make  meaningful comparisons between them. To this end, for each $l\in\mathcal{M}_i[k]$, node $i$ constructs an intermediate quantity $\bar{x}_{i,l}[k]$ by propagating forward the latest estimate it has obtained from node $l$, namely $\hat{x}_l[\phi_{i,l}[k]]$, from time $\phi_{i,l}[k]$ to time $k$ (see Eq. \eqref{eqn:barxil}).\footnote{One can interpret $\bar{x}_{i,l}[k]$ as node $i$'s prediction of node $l$'s state estimate at time $k$, based on its current information.} Note here that $\phi_{i,l}[k]$ is the latest time-step when node $i$ appended node $l$ to $\mathcal{M}_i[k]$. Having constructed the quantities $\bar{x}_{i,l}[k],l\in\mathcal{M}_i[k]$, node $i$ then rejects the highest $f$ and the lowest $f$ of them (i.e., it filters out extreme values), and uses the one that remains, denoted $\bar{x}_i[k]$,  to update $\hat{x}_i[k]$ via \eqref{eqn:algo2filterstep}. 

$\bullet$ \textit{Discussion of Case 2}: The rules in lines 15-18 are essentially the same as those just discussed above for lines 8-12. The key difference between them stems from the manner in which  $\mathcal{M}_i[k]$ is updated. In particular, when $\tau_i[k]=\omega$, or when $\tau_i[k]$ transitions from $\omega$ to a value other than $\omega$, node $i$ is less selective in terms of which nodes to include in $\mathcal{M}_i[k]$; at this point, its main concern is in gathering enough estimates for implementing the filtering step. In contrast, when $\tau_i[k]\neq\omega$, node $i$ carefully examines the freshness-indices of its instantaneous neighbors prior to a potential inclusion in $\mathcal{M}_i[k]$. This is done in two steps. First, if it comes in contact with a node $l$ that already exists in $\mathcal{M}_i[k]$, then it checks whether node $l$ has fresher information to offer than when it was last appended to $\mathcal{M}_i[k]$ (see line 15). If so, node $i$ replaces the entries corresponding to node $l$ with their more recent versions. Next, in line 16, node $i$ considers the set $\mathcal{M}_i[k]\cup \{\mathcal{J}_i[k]\setminus\mathcal{M}_i[k]\}$, and appends/retains only those $2f+1$ nodes that have the lowest freshness-indices.\footnote{In implementing this step, suppose a node $l$ already existing in $\mathcal{M}_i[k]$ gets retained in $\mathcal{M}_i[k]$. Suppose the information concerning node $l$ was stored in the $p$-th components of $\mathbf{d}_i[k],\mathbf{v}_i[k]$ and $\boldsymbol{\phi}_i[k]$. Then, node $i$ continues to store node $l$'s information in the $p$-th components of the above vectors. This is done only to simplify some of the arguments (see proof of Lemma \ref{lemma:algo2dtbounds}). \label{footnt:dstore}} The rationale behind breaking up the ``appending" operation into two steps (lines 15-16) is to avoid the possibility of having multiple appearances of a single node in $\mathcal{M}_i[k]$. Note that for both Case 1 and Case 2, the values of $\mathcal{M}_i,\mathbf{v}_i$ and $\boldsymbol{\phi}_i$ at the beginning of time-step $k+1$ are initialized with their values at the end of time-step $k$ (lines 14 and 18); their values at the end of time-step $k+1$ will naturally depend on the new information available to node $i$ at time $k+1$. Finally, we emphasize that the rules of Algo. \ref{algo:resilient} apply only to the regular nodes; an adversary can deviate from them in arbitrary ways.
\section{Performance Guarantees for Algorithm \ref{algo:resilient}}
In order to state our main result concerning the performance of Algorithm \ref{algo:resilient}, we need to recall the following concepts.
\begin{definition}($r$-\textbf{reachable set}) \cite{rescons} Given a graph $\mathcal{G}=(\mathcal{V,E})$, a set $\mathcal{C} \subseteq \mathcal{V}$, and an integer $r \in \mathbb{N}_{+}$, $\mathcal{C}$ is said to be an \textit{$r$-reachable set} if $\exists i \in \mathcal{C}$ such that $|\mathcal{N}_i \setminus \mathcal{C}| \geq r$. 
\end{definition}

\begin{definition}(\textbf{Strongly} $r$-\textbf{robust graph} \textit{w.r.t.} $\mathcal{S}$) \cite{mitraAuto} Given a graph $\mathcal{G}=(\mathcal{V,E})$, a set $\mathcal{S} \subset \mathcal{V}$, and an integer $r \in \mathbb{N}_{+}$, $\mathcal{G}$ is said to be  \textit{strongly $r$-robust w.r.t. $\mathcal{S}$} if any non-empty subset $\mathcal{C} \subseteq \mathcal{V}\setminus\mathcal{S}$ is $r$-reachable.
\label{defn:strongrobust}
\end{definition}

Based on the above concepts, we now introduce the key graph property that will be of primary importance to us. 
\begin{definition}\textbf{(Joint strong $r$-robustness w.r.t. $\mathcal{S}$)} Given an integer $r \in \mathbb{N}_{+}$, and a set $\mathcal{S} \subset \mathcal{V}$, a sequence of graphs $\{\mathcal{G}[k]\}_{k=0}^{\infty}$ is said to be jointly strongly $r$-robust w.r.t. $\mathcal{S}$ if there exists ${T}\in\mathbb{N}_{+}$ such that $\bigcup\limits_{\tau=kT}^{(k+1)T-1}\mathcal{G}[\tau]$ is strongly $r$-robust w.r.t. $\mathcal{S}, \forall k\in\mathbb{N}$.
\label{def:jointrobust}
\end{definition}

The main result of this section is as follows.
\begin{theorem}
Consider a scalar LTI system of the form \eqref{eqn:system}, and a measurement model of the form \eqref{eqn:Obsmodel}. Let the sequence of communication graphs $\{\mathcal{G}[k]\}_{k=0}^{\infty}$ be jointly strongly $(3f+1)$-robust w.r.t. the source set $\mathcal{S}$. Then, based on Algorithm \ref{algo:resilient}, the estimation error of each node $i\in\mathcal{R}$ can be made to converge to zero exponentially fast at any desired rate $\rho$, despite the actions of any $f$-total Byzantine adversarial set. 
\label{thm:resilient}
\end{theorem}

The proof of the above theorem is deferred to Appendix \ref{sec:proofThm2}. Prior to that, a few remarks are in order.

\begin{remark}
We briefly point out how the above development can be easily generalized to vector dynamical systems. When the system matrix $\mathbf{A}$ contains real, distinct eigenvalues, one can first diagonalize the system. In the new coordinate frame, each state component is a scalar that can be treated analogously as in Algorithm \ref{algo:resilient}. To deal with the case when  $\mathbf{A}$ has arbitrary spectrum, we can first transform $\mathbf{A}$ to its real Jordan canonical form, and then combine ideas from \cite{mitraAuto} and Section \ref{sec:advtimevar}. In each of the above cases, if (i) every unstable mode (eigenvalue) has its own set of source nodes that can observe it, and (ii) joint strong $(3f+1)$-robustness holds w.r.t. each such source set, our results will go through.  
\end{remark}

\begin{remark}
While in Def. \ref{def:jointrobust}, we required the strong-robustness property to be preserved over intervals of constant length $T$, we can easily allow for such intervals to grow linearly as well, in the spirit of condition (C2) in Section \ref{sec:probform}. 
\end{remark}

\begin{remark}
Following the proof of Theorem \ref{thm:resilient} in Section \ref{sec:proofThm2}, it is easy to see that one can obtain finite-time convergence in the adversarial setting as well. To do so, each regular source node $i$ needs to simply choose its observer gain $l_i$ in \eqref{eqn:algo2luen} so that $a-l_ic_i=0$. One can then verify that it would take at most $2(N-|\mathcal{S}|)T+1$ time-steps for the errors of all regular nodes to converge to $0$.
\end{remark}

\begin{remark}
In our prior work \cite{mitraAR}, we studied a restrictive class of time-varying networks, where the restrictiveness arose from the fact that each regular non-source node was constrained ahead of time to only ever listen to certain specific nodes in the network. The main challenge that we address in this paper is to get rid of this constraint, thereby allowing for far more general graph sequences. This, in turn, requires a node to decide in real-time which neighbors to pay attention to (i.e., append to its list) as in Algorithm \ref{algo:resilient}. 
\end{remark}

\begin{remark}
The reason we require joint strong $(3f+1)$-robustness as opposed to joint strong $(2f+1)$-robustness is as follows. Consider a scenario where there are precisely $2f+1$ source nodes, $f$ of whom are adversarial, and there is exactly one non-source node $i$. Suppose the graph sequence is jointly strongly $(2f+1)$-robust w.r.t. $\mathcal{S}$, i.e., each source node will be in a position to transmit information to node $i$ over every interval of the form $[kT,(k+1)T), k\in\mathbb{N}$. Suppose the $f$ adversaries do not transmit at all. Node $i$ will never be able to attribute such a phenomenon to adversarial behavior (since the lack of information from the adversarial nodes could be due to absence of links in the time-varying graph). Thus, $\tau_i[k]$ will equal $\omega$ for all time, and node $i$ will keep running open-loop forever, thereby causing Algorithm  \ref{algo:resilient} to fail.  
\end{remark}

\section{Conclusion}
We proposed a novel approach to the design of distributed observers for LTI systems over time-varying networks. We proved that our algorithm guarantees exponential convergence at any desired rate (including finite-time convergence) under graph-theoretic conditions that are far milder than those existing. In particular, we showed that these results hold even when the inter-communication intervals grow unbounded over time. We then extended our framework to account for the possibility of worst-case adversarial attacks. In terms of future research directions, it would be interesting to explore the performance of our algorithm when the underlying network changes stochastically, as in \cite{alireza}. We believe that the notion of a ``freshness-index", as employed in this paper, should be applicable to other related classes of problems that involve some form of collaborative estimation or inference - investigations along this line also merit attention. 

%--------------------------------
\appendices
%--------------------------------
\section{Proof of Theorem \ref{thm:main1}}
\label{subsec:proofs}
The goal of this section is to prove Theorem \ref{thm:main1}. Before delving into the technical details, we first provide an informal discussion of the main ideas underlying the proof of Theorem \ref{thm:main1}. To this end, let us fix a sub-state $j\in\{1,\ldots,N\}$. The starting point of our analysis is Lemma \ref{lemma:form} which establishes that for any non-source node $i\in\mathcal{V}\setminus\{j\}$, its error in estimation of sub-state $j$ at time-step $k$ can be expressed as a delayed version of the corresponding error of the source node $j$, where the delay is precisely the freshness-index $\mathcal{\tau}^{(j)}_i[k].$ Given this result, we focus on bounding the delay $\mathcal{\tau}^{(j)}_i[k]$ by exploiting the graph connectivity condition (C3). This is achieved in Lemma \ref{lemma:counterinit} where we first establish that $\mathcal{\tau}^{(j)}_i[k]$ gets triggered after a finite period of time, and then show that it can be bounded above by the function $\tilde{g}(k)=2(N-1)g(k)$, where $g(k)$ is as defined in Section \ref{sec:probform}. At this point, we appeal to condition (C2) (which caps the rate of growth of $\tilde{g}(k)$) in designing the observer gain $\mathbf{L}_j$ at node $j$. Specifically, in the proof of Theorem \ref{thm:main1}, we carefully design $\mathbf{L}_j$ such that despite a potentially growing delay, every non-source node $i\in\mathcal{V}\setminus\{j\}$ inherits the same exponential convergence to the true dynamics $\mathbf{z}^{(j)}[k]$ as that achieved by the corresponding source node $j$. With these ideas in place, we first prove a simple result that will be helpful later on; it states that a non-source node for a certain sub-state will always adopt the information of the corresponding source node, whenever it is in a position to do so. 

\begin{lemma}
Consider any sub-state $j$, and suppose that at some time-step $k$, we have  $j\in\mathcal{N}_i[k]$, for some $i\in\mathcal{V}\setminus\{j\}$. Then, based on Algorithm \ref{algo:main}, we have:
\begin{enumerate}
    \item[(i)] If $\tau^{(j)}_i[k]=\omega$, then $j=\argmin_{l\in\mathcal{M}^{(j)}_i[k]} \tau^{(j)}_l[k]$.
    \item[(ii)] If $\tau^{(j)}_i[k]\neq\omega$, then $j=\argmin_{l\in\mathcal{F}^{(j)}_i[k]} \tau^{(j)}_l[k]$.
\end{enumerate}
\label{lemma:sourceuse}
\end{lemma}
\begin{proof}
The result follows from two simple observations that are direct consequences of the rules of Algorithm \ref{algo:main}: (i) $\tau^{(j)}_j[k]=0, \forall k\in\mathbb{N}$, and (ii) for any $i\in\mathcal{V}\setminus\{j\}$, $\tau^{(j)}_i[k]\geq 1$ whenever $\tau^{(j)}_i[k]\neq\omega$. In other words, the source node for a given sub-state has the lowest freshness-index for that sub-state at all time-steps.
\end{proof}
\begin{lemma}
Suppose all nodes employ Algorithm \ref{algo:main}. Consider any sub-state $j$, and suppose that at some time-step $k$, we have  $\tau^{(j)}_i[k]=m$, where $i\in\mathcal{V}\setminus\{j\}$, and $m\in\mathbb{N}_{+}$. Then, there exist nodes $v(\tau)\in\mathcal{V}\setminus\{j\},\tau\in\{k-m+1,\ldots,k\}$, such that the following is true:
\begin{equation}
\resizebox{0.99\hsize}{!}{$
    \hat{\mathbf{z}}^{(j)}_i[k]=\mathbf{A}^m_{jj}\hat{\mathbf{z}}^{(j)}_j[k-m]+\sum\limits_{q=1}^{(j-1)}\hspace{-1mm}\sum\limits_{\tau=(k-m)}^{(k-1)}\hspace{-2.5mm}\mathbf{A}_{jj}^{(k-\tau-1)}\mathbf{A}_{jq}\hat{\mathbf{z}}^{(q)}_{v(\tau+1)}[\tau].$}
    \label{eqn:delayedform}
\end{equation}
\label{lemma:form}
\end{lemma}

\begin{proof}
Consider any sub-state $j$, and suppose that at some time-step $k$, we have  $\tau^{(j)}_i[k]=m$, where $i\in\mathcal{V}\setminus\{j\}$, and $m\in\mathbb{N}_{+}$. Given this scenario, we claim that there exist nodes $v(\tau)\in\mathcal{V}\setminus\{j\},\tau\in\{k-m+1,\ldots,k\}$ such that $v(\tau)$ adopts the information of $v(\tau-1)$ at time $\tau-1$ for sub-state $j$, $\forall \tau\in\{k-m+1,\ldots,k\}$, with $v(k-m)=j$ and $v(k)=i$. As we shall see, establishing this claim readily establishes \eqref{eqn:delayedform}; thus, we first focus on proving the former via induction on $m$. For the base case of induction, suppose $\tau^{(j)}_i[k]=1$ for some $i\in\mathcal{V}\setminus\{j\}$ at some time-step $k$. Based on Algorithm \ref{algo:main} and Lemma \ref{lemma:sourceuse}, note that this is possible if and only if $j\in\mathcal{N}_i[k-1]$. In particular, $v(k)=i$ would then adopt the information of $v(k-1)=j$ at time $k-1$ for sub-state $j$. This establishes the claim when $m=1$. Now fix an integer $r\geq2$, and suppose the claim is true for all $m\in\{1,\ldots,r-1\}$. Suppose $\tau^{(j)}_i[k]=r$ for some $i\in\mathcal{V}\setminus\{j\}$ at some time-step $k$. From Algorithm \ref{algo:main}, observe that this is true if and only if $i$ adopts the information of some node $l\in\mathcal{N}_i[k-1]\cup\{i\}$ at time $k-1$ for sub-state $j$, such that $\tau^{(j)}_l[k-1]=r-1$. Since $r-1\geq1$, it must be that  $l\in\mathcal{V}\setminus\{j\}$; the induction hypothesis thus applies to node $l$. Using this fact, and setting $v(k-1)=l$ completes our inductive proof of the claim. Finally, observe that for any $\tau\in\{k-m+1,\ldots,k\}$, whenever $v(\tau)$ adopts the information of $v(\tau-1)$ at $\tau-1$, the following identity holds based on \eqref{eqn:nonsourcecase1} and \eqref{eqn:nonsourcecase2}: 
\begin{equation}
    \hat{\mathbf{z}}^{(j)}_{v(\tau)}[\tau]=\mathbf{A}_{jj}\hat{\mathbf{z}}^{(j)}_{v(\tau-1)}[\tau-1]+\sum \limits_{q=1}^{(j-1)}\mathbf{A}_{jq}\hat{\mathbf{z}}^{(q)}_{v(\tau)}[\tau-1].
\end{equation}
Using the above identity repeatedly for all $\tau\in\{k-m+1,\ldots,k\}$ with $v(k-m)=j$ and $v(k)=i$,  immediately leads to \eqref{eqn:delayedform}. This completes the proof.
\end{proof}

\begin{lemma}
Suppose the sequence $\{\mathcal{G}[k]\}^{\infty}_{k=0}$ satisfies condition (C3) in Section \ref{sec:probform}. Then, for each sub-state $j$, Algorithm \ref{algo:main} guarantees the following. 
\begin{equation}
    \tau^{(j)}_i[k]\neq\omega,\forall k\geq \sum_{q=0}^{N-2}f(t_q), \forall i\in\mathcal{V}, \hspace{2mm} \textrm{and}
    \label{eqn:counterinit}
\end{equation}
\begin{equation}
   \tau^{(j)}_i[t_{p(N-1)}] \leq\hspace{-3mm} \sum_{q=(p-1)(N-1)}^{p(N-1)-1}\hspace{-6mm}f(t_q), \forall p\in\mathbb{N}_{+},  \forall i\in\mathcal{V}.
   \label{eqn:delaybound}
\end{equation}
\label{lemma:counterinit}
\end{lemma}
\begin{proof}
Fix a sub-state $j$, and notice that both \eqref{eqn:counterinit} and \eqref{eqn:delaybound} hold for the corresponding source node $j$, since $\tau^{(j)}_j[k]=0,\forall k\in \mathbb{N}$. To establish these claims for the remaining nodes, we begin by making the following simple observation that follows directly from \eqref{eqn:indexupdatecase11} and \eqref{eqn:indexupdatecase2}, and applies to every node $i\in\mathcal{V}\setminus\{j\}$:
\begin{equation}
    \tau^{(j)}_i[k+1]\leq\tau^{(j)}_i[k]+1, \hspace{1mm} \textrm{whenever} \hspace{1mm}  \tau^{(j)}_i[k]\neq \omega.
    \label{eqn:indexbound}
\end{equation}
Our immediate goal is to establish \eqref{eqn:delaybound} when $p=1$ and, in the process, establish \eqref{eqn:counterinit}. 
Let $\mathcal{C}^{(j)}_0=\{j\}$, and define:
\begin{equation}
    \mathcal{C}^{(j)}_1\triangleq\{i\in\mathcal{V}\setminus\mathcal{C}^{(j)}_0: \{\bigcup \limits_{\tau=t_0}^{t_1-1}\mathcal{N}_i[\tau]\}\cap\mathcal{C}^{(j)}_0\neq\emptyset\}.
\end{equation}
In words, $\mathcal{C}^{(j)}_1$ represents the set of non-source nodes (for sub-state $j$) that have a direct edge from node $j$ at least once over the interval $[t_0,t_1)$. Based on condition (C3),  $\mathcal{C}^{(j)}_1$ is non-empty (barring the trivial case when $\mathcal{V}=\{j\}$). For each $i\in\mathcal{C}^{(j)}_1$, it must be that $j \in \mathcal{M}^{(j)}_i[\bar{k}]$ for some $\bar{k}\in[t_0,t_1)$. Thus, based on \eqref{eqn:indexupdatecase11} and \eqref{eqn:indexupdatecase2}, it must be that $\tau^{(j)}_i[k]\neq\omega, \forall k\geq t_1=f(t_0), \forall i\in \mathcal{C}^{(j)}_1$. In particular, we note based on \eqref{eqn:indexbound} that $\tau^{(j)}_i[t_1] \leq t_1$, and hence $\tau^{(j)}_i[t_{N-1}] \leq t_{N-1}=\sum_{q=0}^{{N-2}}f(t_q), \forall i\in\mathcal{C}^{(j)}_1$. 
We can keep repeating the above argument by recursively defining the sets $\mathcal{C}^{(j)}_r,1\leq r \leq (N-1)$, as follows:
\begin{equation}
    \mathcal{C}^{(j)}_r\triangleq\{i\in\mathcal{V}\setminus\bigcup \limits_{q=0}^{(r-1)}\mathcal{C}^{(j)}_q:\{\hspace{-2.5mm}\bigcup \limits_{\tau=t_{r-1}}^{t_r-1}\hspace{-1.5mm}\mathcal{N}_i[\tau]\}\cap\{\bigcup \limits_{q=0}^{(r-1)}\mathcal{C}^{(j)}_q\}\neq \emptyset\}.
\end{equation}
We proceed via induction on $r$. Suppose the following is true for all $r\in\{1,\ldots,m-1\}$, where $m\in\{2,\ldots,N-1\}$: $\tau^{(j)}_i[t_r]\neq\omega$ and $\tau^{(j)}_i[t_r]\leq t_r, \forall i\in\bigcup \limits_{q=0}^{r}\mathcal{C}^{(j)}_q$. Now suppose $r=m$. If $\mathcal{V}\setminus\bigcup \limits_{q=0}^{(m-1)}\mathcal{C}^{(j)}_q$ is empty, then we are done establishing \eqref{eqn:counterinit}, and \eqref{eqn:delaybound} for the case when $p=1$. Else, based on condition (C3), it must be that $\mathcal{C}^{(j)}_m$ is non-empty. 
Consider a node $i\in\mathcal{C}^{(j)}_m$. Based on the way $\mathcal{C}^{(j)}_m$ is defined, note that at some time-step $\bar{k}\in[t_{m-1},t_m)$, node $i$ has some neighbor $v$ (say) from the set $\bigcup \limits_{q=0}^{(m-1)}\mathcal{C}^{(j)}_q$. Based on the induction hypothesis and \eqref{eqn:indexbound}, it must be that $\tau^{(j)}_v[\bar{k}]\neq \omega$ and  $\tau^{(j)}_v[\bar{k}] \leq \bar{k}$. At this point, if $\tau^{(j)}_i[\bar{k}]=\omega$, then since $v\in\mathcal{M}^{(j)}_i[\bar{k}]$, node $i$ would update $\tau^{(j)}_i[\bar{k}]$ based on \eqref{eqn:indexupdatecase11}. Else, if $\tau^{(j)}_i[\bar{k}]\neq\omega$, there are two possibilities: (i) $v\in\mathcal{F}^{(j)}_i[\bar{k}]$, implying $\mathcal{F}^{(j)}_i[\bar{k}]\neq\emptyset$; or (ii) $v\notin\mathcal{F}^{(j)}_i[\bar{k}]$, implying $\tau^{(j)}_i[\bar{k}]\leq\tau^{(j)}_v[\bar{k}]\leq \bar{k}$. The above discussion, coupled with the freshness-index update rules for Case 2 of the algorithm (line 5 of Algo. \ref{algo:main}), and \eqref{eqn:indexbound}, imply $\tau^{(j)}_i[t_m]\neq \omega$ and $\tau^{(j)}_i[t_m]\leq t_m$. This completes the induction step. Appealing to \eqref{eqn:indexbound} once again, and noting that $\bigcup_{q=0}^{N-1}\mathcal{C}^{(j)}_q=\mathcal{V}$ and $t_{N-1}=\sum_{q=0}^{N-2}f(t_q)$, establishes $\eqref{eqn:counterinit}$, and \eqref{eqn:delaybound} when $p=1$. 

In order to establish \eqref{eqn:delaybound} for any $p\in\mathbb{N}_{+}$, one can follow a similar line of argument as above to analyze the evolution of the freshness indices over the interval $[t_{(p-1)(N-1)},t_{p(N-1)}]$. In particular, for any $p>1$, we can set $\mathcal{D}^{(j)}_0=\{j\}$, and define the  sets $\mathcal{D}^{(j)}_r, 1\leq r \leq (N-1)$ recursively as follows:
\begin{equation}
\resizebox{0.98\hsize}{!}{$
    \mathcal{D}^{(j)}_r\triangleq\{i\in\mathcal{V}\setminus\bigcup \limits_{q=0}^{(r-1)}\mathcal{D}^{(j)}_q:\{\hspace{-2.5mm}\bigcup \limits_{\tau=t_{h(p,N,r)-1}}^{t_{h(p,N,r)}-1}\hspace{-4.5mm}\mathcal{N}_i[\tau]\}\cap\{\bigcup \limits_{q=0}^{(r-1)}\mathcal{D}^{(j)}_q\}\neq \emptyset\},
    $}
\end{equation}
where $h(p,N,r)=(p-1)(N-1)+r.$ One can then establish that $\tau^{(j)}_i[t_{h(p,N,r)}] \leq \sum_{q=(p-1)(N-1)}^{h(p,N,r)-1}f(t_q),\forall i \in \mathcal{D}^{(j)}_r$, $\forall r\in\{1,\ldots,N-1\}$, via induction.
\end{proof}

We are now in position to prove Theorem \ref{thm:main1}.
\begin{proof} \textbf{(Theorem \ref{thm:main1})} The proof is divided into two parts. In the first part, we describe a procedure for designing the observer gains $\{\mathbf{L}_i\}^{N}_{i=1}$. In the second part, we establish  that our choice of observer gains indeed leads to the desired exponential convergence rate $\rho.$ 

\textbf{Design of the observer gains:} We begin by noting that for each sub-state $j$, one can always find scalars $\beta_j, \gamma_j \geq 1$, such that $\left\Vert{(\mathbf{A}_{jj})}^{k}\right\Vert\leq\beta_j\gamma^k_j, \forall k\in\mathbb{N}$ \cite{horn}.\footnote{We use $\left\Vert\mathbf{A}\right\Vert$ to refer to the induced 2-norm of a matrix $\mathbf{A}$.} Define $\gamma \triangleq \max\limits_{1\leq j \leq N} \gamma_j$. Next, fix a $\bar{\delta}\in (\delta,1)$, where $\delta$ is as in \eqref{eqn:ass_growth}. Given a desired rate of convergence $\rho\in(0,1)$, we now recursively define two sets of positive scalars, namely $\{\rho_j\}^{N}_{j=1}$ and $\{\lambda_j\}^{N}_{j=1}$, starting with $j=N$. With $\lambda_N=\rho$, let $\rho_j, j=N$,  be chosen to satisfy: 
\begin{equation}
    \gamma^{\bar{\delta}}\rho^{1-\bar{\delta}}_j \leq \lambda_j.
    \label{eqn:designrho}
\end{equation}
Having picked $\rho_j\in (0,1)$ to meet the above condition, we set $\lambda_{j-1}$ to be any number in $(0,\rho_j)$, pick $\rho_{j-1}$ to satisfy \eqref{eqn:designrho}, and then repeat this process till we reach $j=1$. Observe that the sets  $\{\rho_j\}^{N}_{j=1}$ and $\{\lambda_j\}^{N}_{j=1}$ as defined above always exist, and satisfy:  $\rho_1<\lambda_1<\rho_2<\lambda_2<\cdots <\lambda_{N-1}<\rho_N<\lambda_N=\rho.$ For each sub-state $j\in\{1,\ldots,N\}$, let the corresponding source node $j$ design the observer gain $\mathbf{L}_j$ (featuring in equation \eqref{eqn:sourceupdate}) in a manner such that the matrix $(\mathbf{A}_{jj}-\mathbf{L}_{j}\mathbf{C}_{jj})$ has distinct real eigenvalues with spectral radius equal to $\rho_j$. Such a choice of $\mathbf{L}_j$ exists as the pair $(\mathbf{A}_{jj},\mathbf{C}_{jj})$ is observable by construction. This completes our design procedure.

\textbf{Convergence analysis:} We first note that there exists a set of positive scalars $\{\alpha_1, \ldots,\alpha_{N}\}$, such that \cite{horn}:
\begin{equation}
    \left\Vert{(\mathbf{A}_{jj}-\mathbf{L}_{j}\mathbf{C}_{jj})}^k\right\Vert \leq \alpha_{j}\rho_{j}^k, \forall k\in\mathbb{N}.
    \label{eqn:schurbound}
\end{equation}
For a particular sub-state $j$, let $\mathbf{e}^{(j)}_i[k]=\hat{\mathbf{z}}^{(j)}_i[k]-\mathbf{z}^{(j)}[k]$. Consider the first sub-state $j=1$, and observe that based on \eqref{eqn:gen_form}, \eqref{eqn:coordinatetransform},    and \eqref{eqn:sourceupdate}, the following is true: $\mathbf{e}^{(1)}_{1}[k+1]=(\mathbf{A}_{11}-\mathbf{L}_{1}\mathbf{C}_{11})\mathbf{e}^{(1)}_{1}[k]$. Thus, we obtain
\begin{equation}
    \mathbf{e}^{(1)}_{1}[k]={(\mathbf{A}_{11}-\mathbf{L}_{1}\mathbf{C}_{11})}^k\mathbf{e}^{(1)}_1[0].
    \label{eqn:mode1rolledout}
\end{equation}
Based on \eqref{eqn:schurbound} and \eqref{eqn:mode1rolledout}, we then have:
\begin{equation}
    \left\Vert\mathbf{e}^{(1)}_1[k]\right\Vert \leq c_1\rho_{1}^k, \forall k\in\mathbb{N},
    \label{eqn:sourcebound1}
\end{equation}
where $c_1\triangleq\alpha_1\left\Vert\mathbf{e}^{(1)}_1[0]\right\Vert$. Given that node 1's error for sub-state 1 decays exponentially as per \eqref{eqn:sourcebound1}, we want to now relate the errors $\mathbf{e}^{(1)}_i[k], i\in\mathcal{V}\setminus\{1\}$ of the non-source nodes (for sub-state 1) to $\mathbf{e}^{(1)}_1[k]$. To this end, consider any $i\in\mathcal{V}\setminus\{1\}$, and note that for any $k\geq t_{N-1}$, Eq. \eqref{eqn:counterinit} in Lemma \ref{lemma:counterinit} implies that $\tau^{(1)}_i[k]\neq\omega$, and hence $\tau^{(1)}_i[k]\in\mathbb{N}_{+}.$ 
Invoking Lemma \ref{lemma:form}, and using the fact that $\mathbf{z}^{(1)}[k]={(\mathbf{A}_{11})}^m\mathbf{z}^{(1)}[k-m], \forall m\in\mathbb{N}$, we then obtain the following $\forall i\in\mathcal{V}\setminus\{1\}$:
\begin{equation}
   \mathbf{e}^{(1)}_i[k]={(\mathbf{A}_{11})}^{\tau^{(1)}_i[k]}\mathbf{e}^{(1)}_1[k-\tau^{(1)}_i[k]], \forall k\geq t_{N-1}.
   \label{eqn:errmode1}
\end{equation}
Our next goal is to bound the delay term $\tau^{(1)}_i[k]$ in the above relation. For this purpose, consider any time-step $k\geq t_{N-1}$, and let $p(k)$ be the largest integer such that $t_{p(k)(N-1)} \leq k$. Then, for any sub-state $j$, and any $i\in\mathcal{V}\setminus\{j\}$, we observe:
\begin{equation}
\begin{aligned}
\tau^{(j)}_i[k] &\overset{(a)}\leq \tau^{(j)}_i[t_{p(k)(N-1)}]+(k-t_{p(k)(N-1)})\\ 
&\overset{(b)}\leq \hspace{-2mm}\sum_{q=(p(k)-1)(N-1)}^{p(k)(N-1)-1}\hspace{-7mm}f(t_q)+(k-t_{p(k)(N-1)})\\
&\overset{(c)}\leq 2(N-1)f(m(k))\overset{(d)}=2(N-1)g(k).
\label{eqn:delaybound2}
\end{aligned}
\end{equation}

In the above inequalities, (a) follows from \eqref{eqn:counterinit} in Lemma \ref{lemma:counterinit} and \eqref{eqn:indexbound}; (b) follows from \eqref{eqn:delaybound} in Lemma \ref{lemma:counterinit}; and (c) follows from the monotonicity of $f(\cdot)$ in condition (C1), and by recalling that $m(k)\triangleq \max\{t_q\in\mathbb{I}:t_q \leq k\}$. Finally, (d) follows by recalling that $g(k)=f(m(k))$. Recalling that $\left\Vert{(\mathbf{A}_{11})}^{k}\right\Vert\leq\beta_1\gamma^k_1$, using the bounds in \eqref{eqn:sourcebound1} and \eqref{eqn:delaybound2}, the fact that $\gamma_1\geq1$ and $\rho_1<1$, and the sub-multiplicative property of the 2-norm, we obtain the following by taking norms on both sides of \eqref{eqn:errmode1}:
\begin{equation}
    \left\Vert\mathbf{e}^{(1)}_i[k]\right\Vert \leq \bar{c}_1{\left(\frac{\gamma_1}{\rho_1}\right)}^{\tilde{g}(k)}\rho_{1}^k, \forall k\geq t_{N-1}, \forall i\in\mathcal{V}\setminus\{1\},
    \label{eqn:boundmode1}
\end{equation}
where $\tilde{g}(k)=2(N-1)g(k)$ and $\bar{c}_1\triangleq c_1 \beta_1$. Based on condition (C3), and our choice of $\bar{\delta}$, observe that there exists $\bar{k}(\bar{\delta})$ such that 
  $\tilde{g}(k)\leq \bar{\delta} k, \forall k \geq \bar{k}(\bar{\delta})$. With $k_1\triangleq\max\{t_{N-1},\bar{k}(\bar{\delta})\}$, we then obtain the following based on \eqref{eqn:designrho} and \eqref{eqn:boundmode1}, for all $k\geq k_1$ and for all $i\in\mathcal{V}\setminus\{1\}$: 
 \begin{equation}
    \left\Vert\mathbf{e}^{(1)}_i[k]\right\Vert \leq \bar{c}_1{\left(\gamma^{\bar{\delta}}_1\rho^{1-\bar{\delta}}_1\right)}^k \leq \bar{c}_1{\left(\gamma^{\bar{\delta}}\rho^{1-\bar{\delta}}_1\right)}^k \leq \bar{c}_1 \lambda^k_1. 
    \label{eqn:boundfinalmd1}
\end{equation}
Note that since $\bar{c}_1 \geq c_1$ and $\lambda_1 > \rho_1$, the above bound  applies to node 1 as well (see equation \eqref{eqn:sourcebound1}). We have thus established that exponential convergence at rate $\lambda_1$ for sub-state 1 holds for each node in the network.

Our aim is to now obtain a bound similar to that in \eqref{eqn:boundfinalmd1} for each sub-state $j\in\{2,\ldots,N\}$. To this end, with $g_{jq}=\left\Vert(\mathbf{A}_{jq}-\mathbf{L}_j\mathbf{C}_{jq})\right\Vert$ and $h_{jq}=\left\Vert\mathbf{A}_{jq}\right\Vert$, let us define the following quantities recursively for $j\in\{2,\ldots,N\}$:
\begin{equation}
    \begin{aligned}
        k_j&\triangleq\frac{k_{j-1}}{(1-\bar{\delta})},\\
        c_j&\triangleq\frac{\alpha_j}{\rho_{j}^{k_{j-1}}}\left(\left\Vert\mathbf{e}^{(j)}_j[k_{j-1}]\right\Vert+\sum \limits_{q=1}^{(j-1)}\frac{g_{jq}\bar{c}_q}{(\rho_j-\lambda_q)}\lambda_{q}^{k_{j-1}}\right),\\
        \bar{c}_j&\triangleq\beta_j\left(c_j+\sum\limits_{q=1}^{(j-1)}\frac{h_{jq}\bar{c}_q}{(\gamma_j-\lambda_q)}\right),
    \end{aligned}
    \label{eqn:defn}
\end{equation}
where $k_1\triangleq\max\{t_{N-1},\bar{k}(\bar{\delta})\}$, $c_1\triangleq\alpha_1\left\Vert\mathbf{e}^{(1)}_1[0]\right\Vert$, and $\bar{c}_1=c_1\beta_1$. Based on the above definitions, we claim that for each sub-state $j\in\{1,\ldots,N\}$, the following is true:
\begin{equation}
    \left\Vert\mathbf{e}^{(j)}_i[k]\right\Vert \leq \bar{c}_j\lambda^k_j, \forall k\geq k_j, \forall i \in \mathcal{V}.
    \label{eqn:errboundmodej}
\end{equation}
We will prove the above claim via induction on the sub-state number $j$. We have already established \eqref{eqn:errboundmodej} for the base case when $j=1$. For $j\geq2$, our strategy will be to first analyze the evolution of $\mathbf{e}^{(j)}_j[k]$ at the source node $j$. From \eqref{eqn:gen_form} and \eqref{eqn:coordinatetransform}, we note that the dynamics of the $j$-th sub-state are coupled with those of the first $j-1$ sub-states. Thus, $\mathbf{e}^{(j)}_j[k]$ will exhibit exponential decay only when the errors for the first $j-1$ sub-states have already started decaying exponentially, with $k_{j-1}$ (as defined in \eqref{eqn:defn}) representing the instant when exponential decay for the $(j-1)$-th sub-state kicks in. Let us now prove that as soon as this happens, the following holds:
\begin{equation}
    \left\Vert\mathbf{e}^{(j)}_j[k]\right\Vert \leq {c}_j\rho^k_j, \forall k\geq k_{j-1}.
    \label{eqn:errsourcej}
\end{equation}
To do so, suppose \eqref{eqn:errboundmodej} holds for all $j\in\{1,\ldots,l-1\}$, where $l\in\{2,\ldots,N\}$. Now let $j=l$ and observe that equations \eqref{eqn:gen_form} and \eqref{eqn:coordinatetransform} yield:
\begin{equation}
\resizebox{0.98\hsize}{!}{$
\begin{aligned}
    &\mathbf{z}^{(l)}[k+1]=\mathbf{A}_{ll}\mathbf{z}^{(l)}[k]+\sum_{q=1}^{(l-1)}\mathbf{A}_{lq}\mathbf{z}^{(q)}[k]\\
    \hspace{-3mm}&=(\mathbf{A}_{ll}-\mathbf{L}_l\mathbf{C}_{ll})\mathbf{z}^{(l)}[k]+\sum \limits_{q=1}^{(l-1)}(\mathbf{A}_{lq}-\mathbf{L}_l\mathbf{C}_{lq})\mathbf{z}^{(q)}[k]+\mathbf{L}_l\mathbf{y}_l[k].
\end{aligned}
$}
\label{eqn:modeldyn}
\end{equation}
Based on the above equation and \eqref{eqn:sourceupdate}, we obtain:
\begin{equation}
\resizebox{0.98\hsize}{!}{$
\mathbf{e}^{(l)}_l[k+1]=(\mathbf{A}_{ll}-\mathbf{L}_l\mathbf{C}_{ll})\mathbf{e}^{(l)}_l[k]+\sum \limits_{q=1}^{(l-1)}(\mathbf{A}_{lq}-\mathbf{L}_l\mathbf{C}_{lq})\mathbf{e}^{(q)}_l[k].
\nonumber
$}
\label{eqn:err_modl}
\end{equation}
Rolling out the above equation starting from $k_{l-1}$ yields:
\begin{equation}
\resizebox{1.02\hsize}{!}{$
    \mathbf{e}^{(l)}_l[k]={F_{ll}}^{(k-k_{l-1})}\mathbf{e}^{(l)}_l[k_{l-1}]+
    \sum \limits_{q=1}^{(l-1)}\hspace{-1mm}\sum  \limits_{\tau=k_{l-1}}^{(k-1)}\hspace{-2mm}{F_{ll}}^{(k-\tau-1)}G_{lq}\mathbf{e}^{(q)}_l[\tau],
    $}
    \label{eqn:modejerrorrolledout}
\end{equation}
where $F_{ll}=(\mathbf{A}_{ll}-\mathbf{L}_l\mathbf{C}_{ll})$, and $G_{lq}=(\mathbf{A}_{lq}-\mathbf{L}_l\mathbf{C}_{lq})$. Taking norms on both sides of the above equation, using the triangle inequality, and the sub-multiplicative property of the two-norm, we obtain:
\begin{equation}
\resizebox{1.01\hsize}{!}{$
    \begin{aligned}
        \left\Vert\mathbf{e}^{(l)}_l[k]\right\Vert &\overset{(a)}{\leq} \alpha_l\rho^k_l\left(\frac{\left\Vert\mathbf{e}^{(l)}_l[k_{l-1}]\right\Vert}{\rho_{l}^{k_{l-1}}}+\frac{1}{\rho_l}\sum_{q=1}^{(l-1)}g_{lq}\hspace{-2mm}\sum  \limits_{\tau=k_{l-1}}^{(k-1)}\hspace{-1mm}\rho^{-\tau}_l\left\Vert\mathbf{e}^{(q)}_l[\tau]\right\Vert\right)\\
        &\overset{(b)}{\leq} \alpha_l\rho^k_l\left(\frac{\left\Vert\mathbf{e}^{(l)}_l[k_{l-1}]\right\Vert}{\rho_{l}^{k_{l-1}}}+\frac{1}{\rho_l}\sum_{q=1}^{(l-1)}g_{lq}\bar{c}_q\hspace{-2mm}\sum  \limits_{\tau=k_{l-1}}^{\infty}\hspace{-1mm}{\left(\frac{\lambda_q}{\rho_l}\right)}^{\tau}\right)\\
        &\overset{(c)}{\leq} c_l\rho^k_l, \forall k \geq k_{l-1}.\\
        \label{eqn:boundsourcemodej}
    \end{aligned}
    $}
\end{equation}
In the above inequalities, (a) follows from \eqref{eqn:schurbound} and by recalling that $g_{lq}=\Vert G_{lq} \Vert$; (b) follows by first applying the induction hypothesis noting that $q\leq (l-1)$ and $\tau\geq k_{l-1}$, and then changing the upper limit of the inner summation (over time); (c) follows by simplifying the preceding inequality using the fact that $\lambda_q < \rho_l, \forall  q \in \{1,\ldots,l-1\}$, and using the definition of $c_l$ in \eqref{eqn:defn}. We have thus obtained a bound on the estimation error of sub-state $l$ at node $l$. To obtain a similar bound for each $i\in\mathcal{V}\setminus\{l\}$, note that equation \eqref{eqn:modeldyn} can be rolled out over time to yield the following for each $m\in\mathbb{N}$:
\begin{equation}
\mathbf{z}^{(l)}[k]=\mathbf{A}_{ll}^{m}\mathbf{z}^{(l)}[k-m]+\sum_{q=1}^{(l-1)}\hspace{-1mm}\sum_{\tau=(k-m)}^{(k-1)}\hspace{-2.5mm}\mathbf{A}_{ll}^{(k-\tau-1)}\mathbf{A}_{lq}\mathbf{z}^{(q)}[\tau].\nonumber
\end{equation}
Leveraging Lemma \ref{lemma:form}, we can then obtain the following error dynamics for a node $i\in\mathcal{V}\setminus\{l\}, \forall k\geq t_{N-1}.$ 
\begin{equation}
\begin{aligned}
    \mathbf{e}^{(l)}_i[k]&={(\mathbf{A}_{ll})}^{\tau^{(l)}_i[k]}\mathbf{e}^{(l)}_l[k-\tau^{(l)}_i[k]]\\
    &+\sum_{q=1}^{(l-1)}\hspace{-1mm}\sum_{\tau=(k-\tau^{(l)}_i[k])}^{(k-1)}\hspace{-2.5mm}\mathbf{A}_{ll}^{(k-\tau-1)}\mathbf{A}_{lq}\mathbf{e}^{(q)}_{v(\tau+1)}[\tau].
    \end{aligned}
    \label{eqn:errornonsourcemodej}
\end{equation}
Based on the above equation, we note that since $\mathbf{A}_{ll}$ can contain unstable eigenvalues, and since $\tau^{(l)}_i[k]$ may grow over time (owing to potentially growing inter-communication intervals), we need the decay in $\mathbf{e}^{(l)}_l[k-\tau^{(l)}_i[k]]$ to dominate the growth due to ${(\mathbf{A}_{ll})}^{\tau^{(l)}_i[k]}$ in order for $\mathbf{e}^{(l)}_i[k]$ to eventually remain bounded. To show that this is indeed the case, we begin by noting the following inequalities that hold when $k\geq k_{l}$:
\begin{equation}
\frac{k_{l-1}}{k}\overset{(a)}{\leq}1-\bar{\delta}\overset{(b)}{\leq}1-\frac{\tilde{g}(k)}{k}\overset{(c)}\leq 
  1-\frac{\tau^{(l)}_i[k]}{k},
  \label{eqn:times}
\end{equation}
where $\tilde{g}(k)=2(N-1)g(k)$. In the above inequalities, (a) follows directly from \eqref{eqn:defn}; (b) follows by noting that $k\geq k_l \implies k \geq \bar{k}(\bar{\delta})$; and (c) follows from \eqref{eqn:delaybound2} and by noting that $k\geq k_l \implies k \geq t_{N-1}.$ We conclude that if $k \geq k_{l}$, then  $k-\tau^{(l)}_i[k]\geq k_{l-1}$. Thus, when $k \geq k_{l}$, at any time-step $ \tau\geq k-\tau^{(l)}_i[k]$, the errors of the first $l-1$ sub-states would exhibit exponential decay based on the induction hypothesis. 
With this in mind, we fix  $i\in\mathcal{V}\setminus\{l\}$, $k\geq k_l$, and bound $\mathbf{e}^{(l)}_i[k]$ by taking norms on both sides of $\eqref{eqn:errornonsourcemodej}$, as follows:  
\begin{equation}
\resizebox{1.01\hsize}{!}{$
    \begin{aligned}
        \left\Vert\mathbf{e}^{(l)}_i[k]\right\Vert &\overset{(a)}{\leq} \beta_l\left(c_l{\left(\frac{\gamma_l}{\rho_l}\right)}^{\tilde{g}(k)}\rho^{k}_l+\gamma^{(k-1)}_l\sum_{q=1}^{(l-1)}h_{lq}\bar{c}_q\hspace{-4.5mm}\sum  \limits_{\tau=(k-\tau^{(l)}_i[k])}^{(k-1)}\hspace{0mm}{\left(\frac{\lambda_q}{\gamma_l}\right)}^{\tau}\right)\\
         &\overset{(b)}{\leq} \beta_l\left(c_l{\left(\frac{\gamma_l}{\rho_l}\right)}^{\tilde{g}(k)}\rho^{k}_l+\gamma^{(k-1)}_l\sum_{q=1}^{(l-1)}h_{lq}\bar{c}_q\hspace{-4.5mm}\sum  \limits_{\tau=(k-\tilde{g}(k))}^{\infty}\hspace{0mm}{\left(\frac{\lambda_q}{\gamma_l}\right)}^{\tau}\right)\\
         &\overset{(c)}{=} \beta_l\left(c_l{\left(\frac{\gamma_l}{\rho_l}\right)}^{\tilde{g}(k)}\rho^{k}_l+\sum_{q=1}^{(l-1)}\frac{h_{lq}\bar{c}_q}{(\gamma_l-\lambda_q)}{\left(\frac{\gamma_l}{\lambda_q}\right)}^{\tilde{g}(k)}\lambda^k_q\right)\\
     &\overset{(d)}{\leq}\bar{c}_l{\left(\gamma^{\bar{\delta}}\rho^{1-\bar{\delta}}_l\right)}^k \overset{(e)}\leq  \bar{c}_l \lambda^k_l.
    \end{aligned}
    $}
    \label{eqn:boundnonsourcemodej}
\end{equation}
In the above steps, (a) follows by first recalling that $\left\Vert{(\mathbf{A}_{ll})}^{k}\right\Vert\leq\beta_l\gamma^k_l, \forall k\in\mathbb{N}$, $h_{lq}=\Vert \mathbf{A}_{lq} \Vert$, $\tilde{g}(k)=2(N-1)g(k)$, and then using the induction hypothesis, equations \eqref{eqn:delaybound2}, \eqref{eqn:errboundmodej}, and \eqref{eqn:boundsourcemodej}, and the facts that $\rho_l < 1, \gamma_l\geq1$; (b) follows by suitably changing the upper and lower limits of the inner summation (over time), a change that is warranted since each summand is non-negative; (c) follows by simplifying the preceding inequality; (d) follows by noting that $\lambda_q < \rho_l,  \forall q\in\{1,\ldots,l-1\}$, using the definition of $\bar{c}_l$ in \eqref{eqn:defn}, and the fact that $\tilde{g}(k)\leq \bar{\delta}k, \forall k \geq k_{l}$; and finally (e) follows from \eqref{eqn:designrho}. This completes the induction step. Let $\mathbf{e}_i[k]=\hat{\mathbf{z}}_i[k]-\mathbf{z}[k]$. Recalling that $\lambda_j \leq \rho, \forall j\in\{1,\ldots,N\}$, we obtain as desired:
\begin{equation}
\resizebox{1\hsize}{!}{$
        \left\Vert\mathbf{e}_i[k]\right\Vert =\sqrt{\sum \limits_{j=1}^{N}{\left\Vert\mathbf{e}^{(j)}_i[k]\right\Vert}^2}
      \leq\left(\sqrt{\sum \limits_{j=1}^{N}\bar{c}^2_j}\right)\rho^{k}, \forall k \geq k_{N}, \forall i\in\mathcal{V}. \nonumber
      $} 
\end{equation}
\end{proof}
%--------------------------------
 \section{Proof of Proposition
 \ref{prop:finitetime}}
 \label{app:proofProp1}
\begin{proof} \textbf{(Proposition \ref{prop:finitetime})} 
For each sub-state $j\in\{1,\ldots,N\}$, let the corresponding source node $j$ design the observer gain $\mathbf{L}_j$ (featuring in equation \eqref{eqn:sourceupdate}) in a manner such that the matrix $(\mathbf{A}_{jj}-\mathbf{L}_{j}\mathbf{C}_{jj})$ has all its eigenvalues at $0$. Such a choice of $\mathbf{L}_j$ exists based on the fact that the pair $(\mathbf{A}_{jj},\mathbf{C}_{jj})$ is observable by construction. Let $n_j=dim(\mathbf{A}_{jj})$. Given the above choice of observer gains, we will prove the result by providing an upper bound on the number of time-steps it takes the error of each node to converge to $\mathbf{0}$. To this end, we first define a sequence $\{\bar{\tau}_j\}_{j=1}^{N}$ of time-steps as follows:
\begin{equation}
    \bar{\tau}_1=\inf\{\tau\in\mathbb{N}_{+},\tau \geq t_{N-1}:k-\tilde{g}(k)\geq n_1, \forall k \geq \tau \},
    \label{eqn:taubar_1}
\end{equation}
where $\tilde{g}(k)=2(N-1)g(k)$, and
\begin{equation}
    \bar{\tau}_j=\inf\{\tau\in\mathbb{N}_{+}:k-\tilde{g}(k)\geq \bar{\tau}_{j-1}+n_j,\forall k \geq \tau \}.
    \label{eqn:taubar_j}
\end{equation}
Based on condition (C3), namely $\limsup_{k\to\infty}g(k)/k=\delta < 1$, observe that $\bar{\tau}_j$ as defined above is finite $\forall j\in\{1,\ldots,N\}$. Next, note that by construction, $(\mathbf{A}_{jj}-\mathbf{L}_j\mathbf{C}_{jj})$ is a nilpotent matrix of index at most $n_j$. Thus, it is easy to see that $\mathbf{e}^{(1)}_1[k]=\mathbf{0}, \forall k\geq n_1$, based on \eqref{eqn:mode1rolledout}. Recall from \eqref{eqn:delaybound2} that for each sub-state $j$, $\tau^{(j)}_i[k] \leq \tilde{g}(k), \forall k\geq t_{N-1}, \forall i\in \mathcal{V}$. From the definition of $\bar{\tau}_1$ in \eqref{eqn:taubar_1}, and equation \eqref{eqn:errmode1}, we immediately obtain that  $\mathbf{e}^{(1)}_i[k]=\mathbf{0},\forall k \geq \bar{\tau}_1, \forall i\in\mathcal{V}$. One can easily generalize this argument to the remaining sub-states by using an inductive reasoning akin to that in the proof of Theorem \ref{thm:main1}. In particular, for any sub-state $j\in\{2,\ldots,N\}$, one can roll out the error dynamics for node $j$ as in \eqref{eqn:modejerrorrolledout}, starting from time-step $\bar{\tau}_{j-1}$. By this time, the induction hypothesis would imply that the estimation errors of all nodes on all sub-states $q\in\{1,\ldots,j-1\}$ have converged to zero. The nilpotentcy of $(\mathbf{A}_{jj}-\mathbf{L}_j\mathbf{C}_{jj})$ would then imply that $\mathbf{e}^{(j)}_j[k]=\mathbf{0}, \forall k \geq \bar{\tau}_{j-1}+n_j$. From the definition of $\bar{\tau}_{j}$ in \eqref{eqn:taubar_j}, and $\eqref{eqn:delaybound2}$, we note that $k\geq \bar{\tau}_j \implies k-\tau^{(j)}_i[k] \geq \bar{\tau}_{j-1}+n_j, \forall i\in\mathcal{V}$. Referring to \eqref{eqn:errornonsourcemodej}, we conclude that $\mathbf{e}^{(j)}_i[k]=\mathbf{0}, \forall k \geq \bar{\tau}_j, \forall i\in\mathcal{V}.$ Based on the above reasoning, the overall error for each node converges to $\mathbf{0}$ in at most $\bar{\tau}_N$ time-steps. 
\end{proof}
%--------------------------------
%--------------------------------
\section{Proof of Theorem \ref{thm:resilient}}
\label{sec:proofThm2}
In this section, we develop the proof of Theorem \ref{thm:resilient}. We begin with the following lemma that characterizes the evolution of the vector of freshness-indices $\mathbf{d}_i[k]$. 

\begin{lemma}
Suppose at any time-step $k\in\mathbb{N}_{+}$, $\tau_i[k] \neq \omega$ for some $i\in\{\mathcal{V}\setminus\mathcal{S}\}\cap\mathcal{R}$. Then, Algorithm \ref{algo:resilient} implies the following.
\begin{align}
\mathbf{d}_i[k+1] &\leq \mathbf{d}_i[k]+\mathbf{1}_{2f+1}, \label{eqn:algo2d_bound}\\
\tau_i[k+1] &\leq \tau_i[k]+1, \label{eqn:algo2tau_bound}
\end{align}
where the first inequality holds component-wise.
\label{lemma:algo2dtbounds}
\end{lemma}
\begin{proof}
Consider a node $i\in\{\mathcal{V}\setminus\mathcal{S}\}\cap\mathcal{R}$, and suppose $\tau_i[k] \neq \omega$. This implies that $|\mathcal{M}_i[k]|=2f+1$, and hence, all entries of the vector $\mathbf{d}_i[k]$ are populated with non-negative integers. Now fix any component $p\in\{1,\ldots,2f+1\}$ of $\mathbf{d}_i[k]$, and suppose that it corresponds to node $l\in\mathcal{M}_i[k]$, i.e., focus on $d_{i,l}[k]$. This entry undergoes the following three operations (in order) prior to the end of time-step $k+1$: (i) It gets incremented by $1$ as per \eqref{eqn:d_update}; (ii) It gets potentially replaced by a value strictly smaller than $d_{i,l}[k]+1$ if node $i$ hears from  node $l$ at time $k+1$ (line 15 of Algo. \ref{algo:resilient}); and (iii) It gets subjected to the operation in line 16 of Algo. \ref{algo:resilient}. Clearly, after operations (i) and (ii), the entry corresponding to node $l$ can increase by at most $1$. At the end of operation (iii), node $l$ either gets retained in $\mathcal{M}_i[k+1]$, or gets replaced. In case it is the former, we clearly have $d_{i,l}[k+1] \leq d_{i,l}[k]+1$, and $d_{i,l}[k+1]$ gets stored in the $p$-th component of $\mathbf{d}_i[k+1]$ (see footnote \ref{footnt:dstore}). If node $l$ gets removed from $\mathcal{M}_i[k+1]$, then given that the $2f+1$ nodes with the lowest freshness-indices populate $\mathcal{M}_i[k+1]$ (line 16 of Algo. \ref{algo:resilient}), it must be that the node replacing $l$ in $\mathcal{M}_i[k+1]$ has freshness-index at most $d_{i,l}[k]+1$ at time $k+1$. Thus, regardless of whether node $l$ gets retained or removed, we have argued that the $p$-th component of $\mathbf{d}_i[k]$ can increase by at most $1$ by the end of time-step $k+1$. Noting that the above argument holds for any component $p$ of $\mathbf{d}_i[k]$ establishes \eqref{eqn:algo2d_bound}; Eq. \eqref{eqn:algo2tau_bound} then follows readily from \eqref{eqn:tau_update} and \eqref{eqn:algo2d_bound}. 
\end{proof}

The next lemma is the analogue of Lemma \ref{lemma:counterinit} for the adversarial case. It tells us that the freshness-index of a regular node gets ``triggered" (i.e., assumes a value other than $\omega$) after a finite period of time, and that it is eventually uniformly bounded above. Here,  uniformity is w.r.t. all graph sequences that are jointly strongly $(3f+1)$-robust w.r.t. $\mathcal{S}$.  
\begin{lemma}
Suppose the conditions in the statement of Theorem \ref{thm:resilient} hold. Then, Algorithm \ref{algo:resilient} guarantees the following. 
\begin{equation}
    \tau_i[k]\neq\omega,\forall k\geq (N-|\mathcal{S}|)T, \forall i\in\mathcal{R}, \hspace{2mm} \textrm{and}
    \label{eqn:algo2counterinit}
\end{equation}
\begin{equation}
   \tau_i[k] \leq 2(N-|\mathcal{S}|)T, \forall k\geq (N-|\mathcal{S}|)T,  \forall i\in\mathcal{R},
   \label{eqn:algo2delaybound}
\end{equation}
where $T$ has the same meaning as in Definition  \ref{def:jointrobust}. 
\label{lemma:algo2counterinit}
\end{lemma}
\begin{proof}
The proof of this result mirrors that of Lemma \ref{lemma:counterinit}; hence, we only sketch the essential details. First, observe that the claims made in the statement of the lemma hold trivially for any $i\in\mathcal{S}\cap\mathcal{R}$, since each such node maintains $\tau_i[k]=0$ for all time. Let us now focus on establishing \eqref{eqn:algo2counterinit} for all regular non-source nodes. To this end, let $\mathcal{C}_0=\mathcal{S}$, and define:
\begin{equation}
    \mathcal{C}_1\triangleq\{i\in\mathcal{V}\setminus\mathcal{C}_0: |\{\bigcup \limits_{\tau=0}^{T-1}\mathcal{N}_i[\tau]\}\cap\mathcal{C}_0|\geq3f+1\}.
\end{equation}
In words, $i\in\mathcal{C}_1$ if and only if node $i$ has at least $3f+1$ neighbors from the source set $\mathcal{C}_0=\mathcal{S}$ over the interval $[0,T)$. Suppose $\mathcal{V}\setminus\mathcal{S} \neq \emptyset$ (else, there is nothing to prove). Then, $\mathcal{C}_1$ must also be non-empty based on the definition of joint strong $(3f+1)$-robustness w.r.t. $\mathcal{S}$. Now consider any $i\in\mathcal{C}_1\cap\mathcal{R}$. Since at most $f$ nodes are adversarial, node $i$ must have heard from at least $2f+1$ regular source nodes (not necessarily all at the same time-step) over the interval $[0,T)$, with each such node reporting a freshness-index of $0$.  Regardless of whether or not node $i$ appends all such nodes to $\mathcal{M}_i$, the fact that it has the opportunity to do so, implies that $\tau_i[T] \neq \omega$. Moreover, given the fact that node $i$ appends a node $l$ to $\mathcal{M}_i[k]$ only if $\tau_l[k] \leq k$, and appealing to \eqref{eqn:algo2tau_bound}, we see that $\tau_i[T]\leq T$. Using \eqref{eqn:algo2tau_bound} again, we have  $\tau_i[(N-|\mathcal{S}|)T]\leq (N-|\mathcal{S}|)T.$
Now define the sets $\mathcal{C}_r, 1\leq r \leq (N-|\mathcal{S}|)$, recursively as follows:
\begin{equation}
    \mathcal{C}_r\triangleq\{i\in\mathcal{V}\setminus\bigcup \limits_{q=0}^{(r-1)}\mathcal{C}_q:|\{\hspace{-2.5mm}\bigcup \limits_{\tau=(r-1)T}^{rT-1}\hspace{-1.5mm}\mathcal{N}_i[\tau]\}\cap\{\bigcup \limits_{q=0}^{(r-1)}\mathcal{C}_q\}|\geq 3f+1\}.
\end{equation}
We can now employ an inductive argument similar to that in Lemma \ref{lemma:counterinit} to establish \eqref{eqn:algo2counterinit}. In particular, proceeding as in the base case when $r=1$, joint strong $(3f+1)$-robustness w.r.t. $\mathcal{S}$ implies that $\mathcal{C}_r\neq \emptyset$, whenever $\mathcal{V}\setminus\bigcup \limits_{q=0}^{(r-1)}\mathcal{C}_q \neq \emptyset$. Since $|\mathcal{A}| \leq f$, any node $i\in\mathcal{C}_r \cap \mathcal{R}$ must then hear from at least $2f+1$ regular nodes in $\bigcup \limits_{q=0}^{(r-1)}\mathcal{C}_q$ over the interval $[(r-1)T, rT)$. Moreover, based on an inductive reasoning, each such node would be a potential candidate for getting appended upon making contact with node $i$. It is easy to then see that for each $i\in\mathcal{C}_r\cap\mathcal{R}$, $\tau_i[rT]\neq \omega$ and, in particular, $\tau_i[rT]\leq rT$ based on  \eqref{eqn:algo2tau_bound}.  This establishes the claim in equation  \eqref{eqn:algo2counterinit}. 

We now turn to establishing the correctness of \eqref{eqn:algo2delaybound}. Let $\bar{T}=(N-|\mathcal{S}|)T.$ We claim $\tau_i[p\bar{T}] \leq \bar{T}, \forall p\in\mathbb{N}_{+}, \forall i\in\mathcal{R}$. Note that the argument for the case when $p=1$ follows from the above discussion, and by using \eqref{eqn:algo2tau_bound}. To prove the claim, it suffices to prove it for the case when $p=2$, since an identical reasoning applies for all $p\geq 2$. For analyzing the case when $p=2$, let $\mathcal{D}_0=\mathcal{S}$, and define the sets $\mathcal{D}_r$, $1\leq r \leq (N-|\mathcal{S}|)$, recursively as follows:
\begin{equation}
    \mathcal{D}_r\triangleq\{i\in\mathcal{V}\setminus\bigcup \limits_{q=0}^{(r-1)}\mathcal{D}_q:|\{\hspace{-5mm}\bigcup \limits_{\tau=\bar{T}+(r-1)T}^{\bar{T}+rT-1}\hspace{-5mm}\mathcal{N}_i[\tau]\}\cap\{\bigcup \limits_{q=0}^{(r-1)}\mathcal{D}_q\}|\geq 3f+1\}.
\end{equation}
Let us first consider the case when $r=1$, and accordingly focus on the interval $[\bar{T}, \bar{T}+T)$. Fix $\tau\in[0,T-1]$, and suppose a regular non-source node $i\in\mathcal{D}_1\cap\mathcal{R}$ hears from $m$ regular source nodes over the interval $[\bar{T},\bar{T}+\tau]$. We then claim that at least $m$ components of $\mathbf{d}_i[\bar{T}+\tau]$ are at most $\tau$. This is immediate when $m=1$ based on \eqref{eqn:algo2d_bound}. In particular, whenever a regular source node transmits to node $i$, given that it always reports a freshness-index of $0$, at least one of the entries of $\mathbf{d}_i$ will take on a value of $0$ at that instant. Now suppose $m=2$, and let node $i$ hear from regular source nodes $v_1$ and $v_2$ at time-steps $\bar{T}+\tau_1$ and $\bar{T}+\tau_2$, respectively, where $0\leq \tau_1 \leq \tau_2 \leq \tau$. Given that node $i$ hears from $v_1$ at $\bar{T}+\tau_1$, based on \eqref{eqn:algo2d_bound}, it must be that at least one component of $\mathbf{d}_i[\bar{T}+\tau_2]$ is at most $\tau_2-\tau_1 \leq \tau_2$. Given this, the fact that node $i$ hears from $v_2$ at $\bar{T}+\tau_2$, and noting that node $i$ appends/retains those $2f+1$ nodes in $\mathcal{M}_i[\bar{T}+\tau_2]$ that have the lowest freshness-indices (see line 16 of Algo. \ref{algo:resilient}), observe that at the end of time-step $\bar{T}+\tau_2$, at least two components of $\mathbf{d}_i[\bar{T}+\tau_2]$ are at most $\tau_2$; from \eqref{eqn:algo2d_bound}, it follows that each of these components in $\mathbf{d}_i[\bar{T}+\tau]$ are at most $\tau$. It is easy to see that the above argument can be generalized for $m\geq 2$. Now based on joint strong $(3f+1)$-robustness w.r.t. $\mathcal{S}$, $\mathcal{D}_1\neq\emptyset$ whenever $\mathcal{V}\setminus\mathcal{D}_0 \neq \emptyset$. Since $|\mathcal{A}|\leq f$, any node $i\in\mathcal{D}_1\cap\mathcal{R}$ is guaranteed to hear from $m=2f+1$ regular source nodes over the interval $[\bar{T},\bar{T}+T-1]$. The above discussion then implies that each component of $\mathbf{d}_i[\bar{T}+T-1]$ is at most $T-1$. It then follows from \eqref{eqn:tau_update} that for each $i\in\mathcal{D}_1\cap\mathcal{R}$, $\tau_i[\bar{T}+T] \leq T$, and hence $\tau_i[2\bar{T}] \leq \bar{T}$ based on \eqref{eqn:algo2tau_bound}. To establish that $\tau_i[2\bar{T}] \leq \bar{T}, \forall i\in\mathcal{R}$, one can proceed via induction to establish that $\tau_i[\bar{T}+rT] \leq rT, \forall i\in\mathcal{D}_r\cap\mathcal{R}$. This can be done using arguments similar to when $r=1$, and hence we omit them. The rest of the proof can be completed as in Lemma \ref{lemma:counterinit}. 
\end{proof}

The next lemma reveals the implication of $\tau_i[k]$ taking on a finite value, based on the rules of Algorithm \ref{algo:resilient}. To state the result, let us define $\Omega^{(r)}[k]\triangleq\{a^r\hat{x}_s[k-r]:s\in\mathcal{S}\cap\mathcal{R}\}$, where $r\in\mathbb{N}$. We will use $Conv(\Omega)$ to denote the convex hull  of a finite set of points $\Omega$. 
\begin{lemma}
Suppose that at any time-step $k\in\mathbb{N}_{+}$, $\tau_i[k]=m$ for some $i\in\{\mathcal{V}\setminus\mathcal{S}\}\cap\mathcal{R}$, where $m\in\mathbb{N}_{+}$. Then, Algorithm  \ref{algo:resilient} implies the following.
\begin{equation}
    \hat{x}_{i}[k] \in Conv\left(\bigcup\limits_{r=1}^{m}\Omega^{(r)}[k]\right).
\label{eqn:inclusion}
\end{equation}
\label{lemma:inclusion}
\end{lemma}
\begin{proof}
We will prove the result via induction on $m$. For the base case of induction with $m=1$, suppose that $\tau_i[k]=1$ for some node $i\in\{\mathcal{V}\setminus\mathcal{S}\}\cap\mathcal{R}$ at some time-step $k$. Based on \eqref{eqn:tau_update}, it must then be the case that $\max_{l\in\mathcal{M}_i[k-1]}d_{i,l}[k-1]=0$. Since each entry of $\mathbf{d}_i[k-1]$ is non-negative by definition, we have $d_{i,l}[k-1]=0, \forall l\in\mathcal{M}_i[k-1]$. Since $\mathcal{M}_i[k-1]$ consists of $2f+1$ distinct nodes, and $|\mathcal{A}| \leq f$, at least $f+1$ entries in $\mathcal{M}_i[k-1]$ correspond to regular nodes, each of which are source nodes (all regular non-source nodes have freshness-indices strictly larger than $1$). Thus, node $i$ must have appended  at least $f+1$ regular source nodes at time $k-1$. For each $l\in\mathcal{M}_i[k-1]\cap\{\mathcal{S}\cap\mathcal{R}\}$, observe that $\bar{x}_{i,l}[k-1]=\hat{x}_l[k-1]$. Given the fact that $|\mathcal{A}| \leq f$, it is easy to see that the following holds based on the filtering operation in line 12 of Algo. \ref{algo:resilient}:
\begin{equation}
\bar{x}_i[k-1] \in Conv(\{\bar{x}_{i,l}[k-1]: l\in \mathcal{M}_{i}[k-1]\cap\mathcal{R}\}).
\label{eqn:convhull}
\end{equation}
Based on the above discussion, we then have that $\bar{x}_i[k-1]\in Conv(\Omega^{(0)}[k-1])$, and hence $\hat{x}_i[k]\in Conv(\Omega^{(1)}[k])$, since $\hat{x}_i[k]=a\bar{x}_i[k-1]$. This completes the proof of the base case. 

To generalize the above result, let us first observe the following identity which holds for any $l\in\mathcal{M}_i[k]\cap\mathcal{R}$, and follows directly from the rules of Algo. \ref{algo:resilient}:\footnote{Essentially, \eqref{eqn:balance} suggests that the difference between node $i$'s internal copy of node $l$'s freshness-index at time $k$, namely $d_{i,l}[k]$, and the actual freshness-index $\tau_{l}[\phi_{i,l}[k]]$ of node $l$ when it was last appended by node $i$, is simply the time that has elapsed since node $l$ was last appended by node $i$, namely $(k-\phi_{i,l}[k])$: this observation follows directly from \eqref{eqn:d_update}.}
\begin{equation}
\tau_{l}[\phi_{i,l}[k]]+(k-\phi_{i,l}[k])=d_{i,l}[k].
\label{eqn:balance}
\end{equation}
Now fix an integer $q\geq2$, and suppose the inclusion in \eqref{eqn:inclusion} holds for all $m\in\{1,\ldots,q-1\}$. Suppose $\tau_i[k]=q$ for some $i\in\{\mathcal{V}\setminus\mathcal{S}\}\cap\mathcal{R}$ at some time-step $k$.
Then, based on $\eqref{eqn:tau_update}$, we have $\max_{l\in\mathcal{M}_i[k-1]}d_{i,l}[k-1] = q-1$. Combining this with \eqref{eqn:balance}, we obtain the following for each $l\in\mathcal{M}_i[k-1]\cap\mathcal{R}$: \begin{equation}
  \tau_{l}[\phi_{i,l}[k-1]]+((k-1)-\phi_{i,l}[k-1]) \leq q-1.  
 \label{eqn:balance1}
\end{equation}
Fix $l\in\mathcal{M}_i[k-1]\cap\mathcal{R}$. The above relation can be now exploited to make certain inferences about the estimate of node $l$ at the time-step $\phi_{i,l}[k-1]$ when it was last appended by node $i$. To see this, consider the case when $l$ is a non-source node. Since $((k-1)-\phi_{i,l}[k-1])$ is non-negative, $\tau_{l}[\phi_{i,l}[k-1]] \leq q-1$ based on \eqref{eqn:balance1}. The induction hypothesis thus applies to node $l$, and we have:
\begin{equation}
    \hat{x}_{l}[\phi_{i,l}[k-1]] \in Conv\left(\bigcup\limits_{r=1}^{\tau_{l}[\phi_{i,l}[k-1]]}\Omega^{(r)}[\phi_{i,l}[k-1]]\right).
    \label{eqn:inc1}
\end{equation}
On the other hand, if $l$ is a source node, then clearly $\hat{x}_l[\phi_{i,l}[k-1]] \in Conv (\Omega^{(0)}[\phi_{i,l}[k-1]]).$ Combining the two cases above, and noting that $\bar{x}_{i,l}[k-1]=a^{g_{i,l}[k-1]}\hat{x}_l[\phi_{i,l}[k-1]]$ based on \eqref{eqn:barxil}, where $g_{i,l}[k-1]=((k-1)-\phi_{i,l}[k-1])$, we have:
\begin{equation}
\begin{aligned}
    \bar{x}_{i,l}[k-1] &\in Conv\left(\bigcup\limits_{r=g_{i,l}[k-1]}^{g_{i,l}[k-1]+\tau_{l}[\phi_{i,l}[k-1]]}\Omega^{(r)}[k-1]\right)\\ &\in Conv\left(\bigcup\limits_{r=0}^{q-1}\Omega^{(r)}[k-1]\right),
\end{aligned}
\label{eqn:inc2}
\end{equation}
where the last inclusion follows by noting that $g_{i,l}[k-1] \geq 0$, and that $g_{i,l}[k-1]+\tau_{l}[\phi_{i,l}[k-1]] \leq q-1$ based on \eqref{eqn:balance1}. 
Appealing to $\eqref{eqn:convhull}$, and using \eqref{eqn:inc2}, we conclude that
\begin{equation}
    \bar{x}_i[k-1] \in Conv\left(\bigcup\limits_{r=0}^{q-1}\Omega^{(r)}[k-1]\right).
\end{equation}
Since $\hat{x}_i[k]=a\bar{x}_i[k-1]$, we then have
\begin{equation}
    \hat{x}_i[k] \in Conv\left(\bigcup\limits_{r=1}^{q}\Omega^{(r)}[k]\right),
\end{equation}
which is precisely the desired conclusion.
\end{proof}

We are now in position to prove Theorem \ref{thm:resilient}. 
\begin{proof} (\textbf{Theorem \ref{thm:resilient}}) Let us begin by noting that given any desired convergence rate $\rho \in (0,1)$, each node $i\in\mathcal{S}\cap\mathcal{R}$ can choose its observer gain $l_i$ to guarantee $|e_i[k]| \leq \alpha \rho^k, \forall k\in \mathbb{N}$, based on the observer $\eqref{eqn:algo2luen}$. Here, $e_i[k]=\hat{x}_i[k]-x[k]$, and $\alpha $ is some suitable constant. 

Now consider any $i\in\{\mathcal{V}\setminus\mathcal{S}\}\cap\mathcal{R}$, and suppose $k\geq (N-|\mathcal{S}|)T$. From Lemma \ref{lemma:algo2counterinit}, we know that $\tau_i[k]\neq\omega$, and hence $\tau_i[k]\in\mathbb{N}_{+}$. Based on Lemma   \ref{lemma:inclusion}, we conclude that there exist non-negative weights $w^{(r)}_{is}[k]$, such that $\sum_{s\in\mathcal{S}\cap\mathcal{R}}\sum_{r=1}^{\tau_i[k]}w^{(r)}_{is}[k]=1$, and 
\begin{equation}
    \hat{x}_i[k]=\sum_{s\in\mathcal{S}\cap\mathcal{R}}\sum_{r=1}^{\tau_i[k]}w^{(r)}_{is}[k] a^r\hat{x}_s[k-r].
    \label{eqn:convform}
\end{equation}
Based on the convexity of the weights $w^{(r)}_{is}[k]$, note that ${x}[k]=\sum_{s\in\mathcal{S}\cap\mathcal{R}}\sum_{r=1}^{\tau_i[k]}w^{(r)}_{is}[k] a^r{x}[k-r]$. Using this and \eqref{eqn:convform} yields:
\begin{equation}
    {e}_i[k]=\sum_{s\in\mathcal{S}\cap\mathcal{R}}\sum_{r=1}^{\tau_i[k]}w^{(r)}_{is}[k] a^r{e}_s[k-r].
\end{equation}
Taking norms on both sides of the above equation, we obtain:
\begin{align}
    |e_i[k]| &\leq \alpha \sum_{s\in\mathcal{S}\cap\mathcal{R}}\sum_{r=1}^{\tau_i[k]}w^{(r)}_{is}[k] {\left(\frac{|a|}{\rho}\right)}^{r} \rho^k\\
    &\leq \alpha {\left(\frac{|a|}{\rho}\right)}^{\tau_i[k]} \rho^k\\
    &\leq \alpha {\left(\frac{|a|}{\rho}\right)}^{2(N-|\mathcal{S}|)T} \rho^k.
\end{align}
In the above steps, we have assumed $|a| \geq 1$ (to avoid trivialities) and exploited the convexity of the weights $w^{(r)}_{is}[k]$ to arrive at the second inequality, and used \eqref{eqn:algo2delaybound} to arrive at the final inequality. This concludes the proof. 
\end{proof}
\bibliographystyle{IEEEtran}
\bibliography{refs}

\begin{thebibliography}{10}
\providecommand{\url}[1]{#1}
\csname url@samestyle\endcsname
\providecommand{\newblock}{\relax}
\providecommand{\bibinfo}[2]{#2}
\providecommand{\BIBentrySTDinterwordspacing}{\spaceskip=0pt\relax}
\providecommand{\BIBentryALTinterwordstretchfactor}{4}
\providecommand{\BIBentryALTinterwordspacing}{\spaceskip=\fontdimen2\font plus
\BIBentryALTinterwordstretchfactor\fontdimen3\font minus
  \fontdimen4\font\relax}
\providecommand{\BIBforeignlanguage}[2]{{%
\expandafter\ifx\csname l@#1\endcsname\relax
\typeout{** WARNING: IEEEtran.bst: No hyphenation pattern has been}%
\typeout{** loaded for the language `#1'. Using the pattern for}%
\typeout{** the default language instead.}%
\else
\language=\csname l@#1\endcsname
\fi
#2}}
\providecommand{\BIBdecl}{\relax}
\BIBdecl

\bibitem{chen}
C.-T. Chen, \emph{Linear system theory and design}.\hskip 1em plus 0.5em minus
  0.4em\relax Oxford University Press, Inc., 1998.

\bibitem{dist3}
U.~Khan, S.~Kar, A.~Jadbabaie, and J.~M. Moura, ``On connectivity,
  observability, and stability in distributed estimation,'' in \emph{{P}roc. of
  the 49th IEEE Conference on Decision and Control}, 2010, pp. 6639--6644.

\bibitem{ugrinov}
V.~Ugrinovskii, ``Conditions for detectability in distributed consensus-based
  observer networks,'' \emph{IEEE Trans. on Autom. Control}, vol.~58, no.~10,
  pp. 2659--2664, 2013.

\bibitem{kim}
T.~Kim, H.~Shim, and D.~D. Cho, ``Distributed luenberger observer design,'' in
  \emph{{P}roc. of the 55th IEEE Decision and Control Conference}, 2016, pp.
  6928--6933.

\bibitem{martins}
S.~Park and N.~C. Martins, ``Design of distributed {LTI} observers for state
  omniscience,'' \emph{IEEE Trans. on Autom. Control}, vol.~62, no.~2, pp.
  561--576, 2017.

\bibitem{mitraTAC}
A.~Mitra and S.~Sundaram, ``Distributed observers for {LTI} systems,''
  \emph{IEEE Trans. on Autom. Control}, vol.~63, no.~11, pp. 3689--3704, 2018.

\bibitem{wang}
L.~Wang and A.~S. Morse, ``A distributed observer for a time-invariant linear
  system,'' \emph{IEEE Trans. on Autom. Control}, vol.~63, no.~7, 2018.

\bibitem{han}
W.~Han, H.~L. Trentelman, Z.~Wang, and Y.~Shen, ``A simple approach to
  distributed observer design for linear systems,'' \emph{IEEE Trans. on Autom.
  Control}, vol.~64, no.~1, pp. 329--336, 2019.

\bibitem{rego}
F.~F. Rego, A.~P. Aguiar, A.~M. Pascoal, and C.~N. Jones, ``A design method for
  distributed {L}uenberger observers,'' in \emph{{P}roc. of the 56th IEEE
  Conference on Decision and Control}, 2017, pp. 3374 -- 3379.

\bibitem{nozal}
{\'A}.~R. del Nozal, P.~Mill{\'a}n, L.~Orihuela, A.~Seuret, and L.~Zaccarian,
  ``Distributed estimation based on multi-hop subspace decomposition,''
  \emph{Automatica}, vol.~99, pp. 213--220, 2019.

\bibitem{wang2}
L.~Wang, A.~Morse, D.~Fullmer, and J.~Liu, ``A hybrid observer for a
  distributed linear system with a changing neighbor graph,'' in \emph{{P}roc.
  of the 56th IEEE Conf. on Decision and Control}, 2017, pp. 1024--1029.

\bibitem{ren}
S.~Wang and W.~Ren, ``On the convergence conditions of distributed dynamic
  state estimation using sensor networks: A unified framework,'' \emph{IEEE
  Trans. on Cont. Sys. Tech.}, vol.~26, no.~4, pp. 1300--1316, 2018.

\bibitem{env1}
K.~M. Lynch, I.~B. Schwartz, P.~Yang, and R.~A. Freeman, ``Decentralized
  environmental modeling by mobile sensor networks,'' \emph{IEEE transactions
  on robotics}, vol.~24, no.~3, pp. 710--724, 2008.

\bibitem{env2}
R.~Graham and J.~Cort{\'e}s, ``Adaptive information collection by robotic
  sensor networks for spatial estimation,'' \emph{IEEE Transactions on
  Automatic Control}, vol.~57, no.~6, pp. 1404--1419, 2011.

\bibitem{mitraAR}
A.~Mitra, J.~A. Richards, S.~Bagchi, and S.~Sundaram, ``Resilient distributed
  state estimation with mobile agents: overcoming {B}yzantine adversaries,
  communication losses, and intermittent measurements,'' \emph{Autonomous
  Robots}, vol.~43, no.~3, pp. 743--768, 2019.

\bibitem{mou}
S.~Mou, J.~Liu, and A.~S. Morse, ``A distributed algorithm for solving a linear
  algebraic equation,'' \emph{IEEE Trans. on Autom. Control}, vol.~60, no.~11,
  pp. 2863--2878, 2015.

\bibitem{jadcons}
A.~Jadbabaie, J.~Lin, and A.~Morse, ``Coordination of groups of mobile
  autonomous agents using nearest neighbor rules,'' \emph{IEEE Trans. on Autom.
  Control}, vol.~48, no.~6, pp. 988--1001, 2003.

\bibitem{nedic}
A.~Nedic and A.~Ozdaglar, ``Distributed subgradient methods for multi-agent
  optimization,'' \emph{IEEE Trans. on Autom. Control}, vol.~54, no.~1, p.~48,
  2009.

\bibitem{lin}
H.~Lin and P.~J. Antsaklis, ``Stability and stabilizability of switched linear
  systems: a survey of recent results,'' \emph{IEEE Trans. on Autom. control},
  vol.~54, no.~2, pp. 308--322, 2009.

\bibitem{deghat}
M.~Deghat, V.~Ugrinovskii, I.~Shames, and C.~Langbort, ``Detection and
  mitigation of biasing attacks on distributed estimation networks,''
  \emph{Automatica}, vol.~99, pp. 369--381, 2019.

\bibitem{junsoo}
J.~Kim, J.~G. Lee, C.~Lee, H.~Shim, and J.~H. Seo, ``Local identification of
  sensor attack and distributed resilient state estimation for linear
  systems,'' in \emph{{P}roc. of the 57th IEEE Conference on Decision and
  Control}, 2018, pp. 2056--2061.

\bibitem{he}
X.~He, X.~Ren, H.~Sandberg, and K.~H. Johansson, ``Secure distributed filtering
  for unstable dynamics under compromised observations,''
  \emph{arXiv:1903.07345}, 2019.

\bibitem{mitraAuto}
A.~Mitra and S.~Sundaram, ``Byzantine-resilient distributed observers for {LTI}
  systems,'' \emph{Automatica}, vol. 108, p. 108487, 2019.

\bibitem{kaul}
S.~Kaul, R.~Yates, and M.~Gruteser, ``Real-time status: {H}ow often should one
  update?'' in \emph{IEEE INFOCOM}, 2012, pp. 2731--2735.

\bibitem{costa}
M.~Costa, M.~Codreanu, and A.~Ephremides, ``On the age of information in status
  update systems with packet management,'' \emph{IEEE Transactions on
  Information Theory}, vol.~62, no.~4, pp. 1897--1910, 2016.

\bibitem{talak}
R.~Talak, S.~Karaman, and E.~Modiano, ``Minimizing age-of-information in
  multi-hop wireless networks,'' in \emph{{P}roc. Annual Allerton Conf. on
  Comm., Control, and Computing}, 2017, pp. 486--493.

\bibitem{ACC19}
A.~Mitra, J.~A. Richards, S.~Bagchi, and S.~Sundaram, ``Finite-time distributed
  state estimation over time-varying graphs: Exploiting the
  age-of-information,'' in \emph{{P}roc. of the American Control
  Conference}.\hskip 1em plus 0.5em minus 0.4em\relax IEEE, 2019, pp.
  4006--4011.

\bibitem{Byz}
D.~Dolev, N.~A. Lynch, S.~S. Pinter, E.~W. Stark, and W.~E. Weihl, ``Reaching
  approximate agreement in the presence of faults,'' \emph{Journal of the ACM
  (JACM)}, vol.~33, no.~3, pp. 499--516, 1986.

\bibitem{rescons}
H.~J. LeBlanc, H.~Zhang, X.~Koutsoukos, and S.~Sundaram, ``Resilient asymptotic
  consensus in robust networks,'' \emph{IEEE Journal on Selected Areas in
  Comm.}, vol.~31, no.~4, pp. 766--781, 2013.

\bibitem{vaidyacons}
N.~H. Vaidya, L.~Tseng, and G.~Liang, ``Iterative approximate {B}yzantine
  consensus in arbitrary directed graphs,'' in \emph{Proc. of the ACM Symp. on
  Principles of Distributed Comp.}, 2012, pp. 365--374.

\bibitem{alireza}
A.~Tahbaz-Salehi and A.~Jadbabaie, ``A necessary and sufficient condition for
  consensus over random networks,'' \emph{IEEE Trans. on Autom. Control},
  vol.~53, no.~3, pp. 791--795, 2008.

\bibitem{horn}
R.~A. Horn, R.~A. Horn, and C.~R. Johnson, \emph{Matrix analysis}.\hskip 1em
  plus 0.5em minus 0.4em\relax Cambridge university press, 1990.

\end{thebibliography}
\end{document}